\documentclass[12pt,journal,onecolumn]{IEEEtran}
\usepackage{times} 
\usepackage{epsfig}
\usepackage{pgf,tikz}
\usetikzlibrary{arrows}
\usepackage{algpseudocode}
\usepackage{pgf,tikz}
\usepackage{mathrsfs}
\usepackage{bm}

\usepackage{algorithm,algorithmicx,subfigure, latexsym, amsmath, amsfonts, amssymb, cite,amsthm}
\newtheorem{lemma}{Lemma}
\newtheorem{theorem}{Theorem}
\newtheorem{corollary}{Corollary}
\newtheorem{remark}{Remark}
\newtheorem{definition}{Definition}

\newtheorem{hypothesis}{Hypothesis}

\definecolor{qqqqcc}{rgb}{0.,0.,0.8}
\definecolor{ttttff}{rgb}{0.2,0.2,1.}
\definecolor{xdxdff}{rgb}{0.49,0.49,1.}
\definecolor{zzttqq}{rgb}{0.6,0.2,0.}
\definecolor{qqzzqq}{rgb}{0.,0.6,0.}
\definecolor{ttzzqq}{rgb}{0.2,0.6,0.}

\begin{document}
\title{\LARGE Subgradient-Free Stochastic Optimization Algorithm for Non-smooth Convex Functions over Time-Varying Networks}
\author{Yinghui~Wang, Wenxiao~Zhao, Yiguang~Hong and Mohsen Zamani
\thanks{Y. Wang, W. Zhao, and Y. Hong are with School of Mathematical Sciences, University of Chinese Academy of Sciences, and Key Laboratory of Systems and Control, Academy of Mathematics and Systems Science, Chinese Academy of Sciences, Beijing,  China. e-mail: (wangyinghuisdu@163.com, wxzhao@amss.ac.cn, yghong@iss.ac.cn).}
\thanks{Mohsen Zamani is with School of Electrical Engineering and Computing, University of Newcastle, Callaghan, NSW 2308, Australia.  email:(mohsen.zamani@newcastle.edu.au)}

\thanks{Manuscript received June 22, 2018.}}


\maketitle

\begin{abstract}
In this paper we consider a distributed stochastic optimization problem without the gradient/subgradient information for the local objective functions, subject to local convex constraints.  The objective functions may be non-smooth and observed with stochastic noises, and the network for the distributed design is time-varying. By adding the stochastic dithers into the local objective functions and constructing the randomized differences motivated by the Kiefer-Wolfowitz algorithm, we propose a distributed subgradient-free algorithm to find the global minimizer with local observations. Moreover, we prove that the consensus of estimates and global minimization can be achieved with probability one over the time-varying network, and then obtain the convergence rate of the mean average of estimates as well.  Finally, we give a numerical example to illustrate the effectiveness of the proposed algorithm.
\end{abstract}

\IEEEpeerreviewmaketitle


\section{Introduction}\label{sec:intro}

Many problems arising from control design, signal processing, and data analysis often encounter the optimization of a global objective function consisting of a sum of convex functions in a network environment. This scenario can be seen in applications such as multi-agent coordination, sensor networks and computation of computer clusters \cite{forero2011distributed, lesser2003distributed, wang2018distributed, xiao2007distributed, yi2016initialization}. It is common that a single node/agent in the network corresponds to a single convex function within the objective function. And there are major concerns in relation to such large scale systems such as privacy protection, energy consumption and computation cost. Distributed algorithms provide flexibility and solutions for handling some of these issues and hence become more popular tools to solve the optimization problems in large-scale networks. In other words, distributed optimization designs have been widely studied in recent years with the nodes/agents only exchanging information with their intermediate neighbors.  Over the past few years there have been considerable works, including distributed (stochastic) subgradient methods \cite{nedic2009distributed, ram2010distributed,  yuan2016convergence}, distributed primal-dual subgradient methods \cite{duchi2012dual, lei2016primal}, distributed alternating direction method of multipliers (ADMM) \cite{makhdoumi2017convergence, shi2014linear}, distributed accelerated gradient methods \cite{jakovetic2014fast} for various distributed optimization problems.

Most of the existing results naturally require first-order or second-order gradient/subgradient information of local objective functions corresponding to nodes/agents within the network. However, obtaining gradient/subgradient information is, sometimes, computationally costly and even impracticable for some cases \cite{conn2009introduction, nesterov2011random}. For example, in learning phase of deep neural networks \cite{deng2014deep}, the connection between objective functions and decision variables is too complicated so that we can hardly derive the explicit form of the first-order gradient/subgradient.  Therefore, it is natural to ask how to develop gradient/subgradient-free algorithms, also called  derivative-free algorithms or zeroth-order algorithms in optimization literatures.
In fact, approaches related to  gradient/subgradient-free designs can be found in \cite{chen1999kiefer, conn2009introduction, duchi2013optimal, gao2014on, kiefer1952stochastic, koronacki1975random, nelder1965simplex, nesterov2011random}
and references listed therein. The authors of monographs \cite{conn2009introduction,  nesterov2011random} summarized several classes of derivative-free methods, including both deterministic and stochastic cases. In addition, there were some other algorithms for solving convex optimization using  derivative-free or zeroth-order information, including  the zeroth-order mirror decent algorithm \cite{duchi2013optimal}, the zeroth-order ADMM algorithm \cite{gao2014on},  the  Kiefer-Wolfowitz (KW) algorithm \cite{kiefer1952stochastic}, and the Nelder-Mead algorithm  \cite{nelder1965simplex}. For the Nelder-Mead algorithm, introduced in  \cite{nelder1965simplex}, to the best of our knowledge, theoretical properties are still under investigation. On the other hand, Kiefer-Wolfowitz (KW) algorithms, first proposed in \cite{kiefer1952stochastic} and then further discussed in various areas \cite{chen1999kiefer,koronacki1975random}, introduced stochastic dithers at the points with the values of the objective function to be observed and then constructed randomized differences served as gradients in algorithms. However, these approaches are basically centralized when dealing with the derivative-free optimization problem, and cannot be implemented directly in a distributed derivative-free setting.

Therefore, how to develop distributed gradient/subgradient-free algorithms for distributed optimization problems is an important problem, though only few results just appeared in the past several years \cite{anit2018distributed, hajinezhad2017zeroth, yuan2015randomized, yuan2015zeroth}. Note that, in \cite{hajinezhad2017zeroth, yuan2015randomized, yuan2015zeroth}, the Gaussian approximation was exploited to approximate original (maybe non-smooth) objective functions, and randomized differences were adopted to replace the gradient information to construct these distributed algorithms. In these references, the optimization error was characterized by the gap between Gaussian approximation functions and the original ones. However, the design of distributed gradient/subgradient-free algorithms for the global minimization of non-smooth objective functions over a time-varying network has not been fully investigated.

In this paper, we consider subgradient-free algorithm design for a distributed constrained optimization problem with local non-smooth objective functions over a time-varying network, and observed with stochastic noises. The contributions of this paper are summarized as follows:
\begin{itemize}
\item[(a)]Different from existing distributed algorithms which often require first-order or second-order gradient/subgradient information of local objective functions, we propose a class of subgradient-free approaches for distributed  convex optimization problem with non-smooth local objective functions over time-varying networks.  Note that the distributed KW algorithm for smooth strongly local convex objective functions \cite{anit2018distributed}  and centralized KW algorithms   \cite{chen1999kiefer, kiefer1952stochastic,koronacki1975random} all require the Lipschitz continuity of the gradients of (local) objective smooth functions.

\item[(b)] We design a class of distributed optimization algorithms by a way to combining conventional KW ideas and consensus based algorithms, which is a different technique compared to other existing gradient/subgradient-free algorithms  c.f., \cite{anit2018distributed, hajinezhad2017zeroth, pang2017distributed, yuan2015randomized, yuan2015zeroth}.  We prove the consensus of estimates and achievement of the global minimization with probability one. We further establish the mean-square convergence rate for the estimates as well.  The existing algorithms given in \cite{hajinezhad2017zeroth, pang2017distributed, yuan2015randomized, yuan2015zeroth} utilize Gaussian smooth approximation functions of the original objective functions and the optimization error is characterized in expectation. However, our proposed algorithms deal with the possible non-smooth objective functions directly and converge explicitly to the global minimizer almost surely.  In addition, due to the complications arising from considering optimization constraints and time-varying network topologies, theoretical analysis given in this paper goes beyond that of the (distributed) KW algorithms in \cite{anit2018distributed, chen1999kiefer, kiefer1952stochastic,koronacki1975random}.

\item[(c)] Our proposed algorithms belongs to the distributed stochastic optimization algorithm category, which is an important research area (referring to \cite{ram2010distributed,yuan2016convergence}).   We establish the almost sure convergence and then mean square convergence rate for the global minimization of the proposed  algorithms, and obtain essential properties of distributed stochastic optimization algorithms with diminishing step-size, by extending results of conventional gradient-free KW algorithms, given in e.g., \cite{chen1999kiefer, kiefer1952stochastic,koronacki1975random}, to subgradient-free cases.  Also, we obtain a mean-square convergence rate for the distributed subgradient-free optimization.
\end{itemize}

The rest of the paper is organized as follows. Mathematical preliminaries and problem formulation is formulated in Section \ref{sec2}. Then a class of distributed subgradient-free algorithms and related hypotheses are introduced in Section \ref{sec3}. The proposed algorithms are fully analyzed in Section \ref{sec4}.  Following that, an numerical example is given in Section \ref{sec5}.   Finally, some concluding remarks are addressed in Section \ref{sec6}.

\textbf{Notations}. Denote $\mathcal{R}$ and $\mathcal{R}^{m}$ as the $1$-dimensional and $m$-dimensional Euclidean spaces, respectively. The vectors in this paper are viewed as column vectors unless otherwise stated. For a given vector $x\in \mathcal{R}^{m}$, we denote its transpose by $x^{\top}$. The inner product of vectors $x$ and $y$ is given by $\langle x, y\rangle=x^{\top}y$. The Euclidean norm of $x$ is denoted by $\|x\|_{2}$.  For a function $f(\cdot)$, denote $\partial f(x)$ as its subgradient at $x$ and $dom(f)$ as its function domain.

\section{Mathematical Preliminaries and Problem Formulation}\label{sec2}

In this section, we first introduce mathematical preliminaries about non-smooth analysis and probability theory, and then we give the problem formulation of this paper.
\subsection{Mathematical Preliminaries}\label{sec2.1}

\subsubsection*{Non-smooth analysis}
We first briefly summarize some results on non-smooth analysis \cite{clarke1998nonsmooth,hiriart2012fundamentals} to be used later in this paper.

\begin{definition}\label{Def1}
\cite{hiriart2012fundamentals} [Sub-gradient]
The vector-valued function $\partial f (x)\in \mathcal{R}^{m}$ is called the subgradient of a non-smooth convex function $f(x):\mathcal{R}^{m}\rightarrow \mathcal{R}$ if for any  $x,y\in dom(f)$, the following inequality holds:
\begin{align*}
f(x)-f(y)-\big \langle \partial f(y), x-y\big\rangle\geqslant 0.
\end{align*}
\end{definition}

The following lemma is important for the non-smooth analysis of the proposed algorithms.

\begin{lemma}\label{Lem1}\cite{clarke1998nonsmooth} (Lebourg's Mean Value Theorem) Let $x,y\in X$ and suppose $f(x):\mathcal{R}^{m}\rightarrow \mathcal{R}$ is Lipschitz on an open set containing the line segment $[x,y]$. Then, there exists a point $u\in (x,y )$ such that
\begin{align*}
f(x)-f(y)\in \langle \partial f(u), x-y \rangle.
\end{align*}
\end{lemma}

\subsubsection*{Euclidean norm inequalities}
The following inequalities holds for the Euclidean norm:
\begin{lemma}(\cite{ram2010distributed})\label{Prop1}
Let $x^{1}, x^{2}, \ldots, x^{n}$ be vectors in $\mathcal{R}^{m}$. Then
\begin{align*}
\sum_{i=1}^{n}\Big\|x^{i}-\frac{1}{n}\sum_{j=1}^{n}x^{j}\Big\|_{2}^{2}\leqslant \sum_{i=1}^{n}\Big\|x^{i}-x\Big\|_{2}^{2}, \quad \forall x\in \mathcal{R}^{m}.
\end{align*}
\end{lemma}

Let us denote $P_{X}(x)$ as the projection of $x$ onto set $X$, i.e., $P_{X}(x)=\arg\min_{y\in X} \big\|x-y \big\|_{2}$, where $X$ is a closed convex set in $\mathcal{R}^{m}$. Then the following result holds true for the projection operator:
\begin{lemma}\label{Lem2}
\cite{hiriart2012fundamentals, nedic2010constrained} Let $X$ be a a closed convex set in $\mathcal{R}^{m}$. Then for any $x\in \mathcal{R}^{m}$, it holds that
\begin{itemize}
\item[(a)]$\big\langle x-P_{X}(x), y-P_{X}(x)\big\rangle \leqslant 0$,  for all $y\in X$
\item[(b)]$\big\|P_{X}(x)-P_{X}(y) \|_{2} \leqslant \big\|x-y \big\|_{2}$, for all $x,y\in \mathcal{R}^{m}$.
\item[(c)]$\big\langle x-y, P_{X}(y)-P_{X}(x)\big\rangle \leqslant -\big\|P_{X}(x)-P_{X}(y) \big\|_{2}^{2}$, for all $y\in \mathcal{R}^{m}$.
\item[(d)] $\big\|x-P_{X}(x) \big\|_{2}^{2} +\big \|y-P_{X}(x)\big\|_{2}^{2} \leqslant \big\| x-y\big\|_{2}^{2}$, for any $y\in X$.
\end{itemize}
\end{lemma}

\subsubsection*{Probability theory}
Denote $(\Omega,\mathcal{F}, \mathbb{P})$ as the basic probability space, where $\Omega$ the whole event space, $\mathcal{F}$ the $\sigma$-algebra on $\Omega$, and $\mathbb{P}$ the probability measure on $(\Omega,\mathcal{F})$.  Next, we give definitions of convergence in probability theory and a lemma of the convergence of super-martingales.

\begin{definition}\label{Def2}
 \cite{durrett2010probability}[Convergence  in $(\Omega,\mathcal{F}, \mathbb{P})$]
 \begin{itemize}
 \item[(a)] Let $x_{1},x_{2},\ldots,x_{k}\ldots$ be a sequence of random variables in $(\Omega,\mathcal{F}, \mathbb{P})$. If $\mathbb{P}(x_{k}\rightarrow x)=1$, we say that $x_{k}$ converges $x$ almost surely (a. s.).
 \item[(b)] Let $x_{1},x_{2},\ldots,x_{k}\ldots$ be a sequence of random variables in $(\Omega,\mathcal{F}, \mathbb{P})$. If  $\mathbb{E}\|x_{k}-x\|^{p}\rightarrow 0$, we say that $x_{k}$ converges to $x$ in $L^{p}$.
 \end{itemize}
\end{definition}

\begin{lemma}\label{Lem3}(\cite{polyak1987introduction}) Denote $(\Omega,\mathcal{F}, \mathbb{P})$ as the basic probability space and $\{ F_{k}\}_{k\geq1}$ as a sequence of increasing sub-$\sigma$-algebras on $\mathcal{F}$. $\{h_k\}_{k\geq1}$, $\{v_k\}_{k\geq1}$ and $\{w_k\}_{k\geq1}$ are scalar variable sequences such that $h_k$, $v_k$ and $w_k$ are $F_k$-measurable for each $k$. Both $\{v_k\}_{k\geq1}$ and $\{w_k\}_{k\geq1}$ are nonnegative and  $\sum_{k=1}^{\infty} w_k <\infty$. Furthermore, $\{h_k\}_{k\geq1}$ is bounded from below uniformly. If the following inequality holds with probability one,
\begin{align*}
\mathbb{E}[h_{k+1}| F_k] \leqslant (1+\eta_k) h_k-v_k+ w_k,\;\; \forall k\geqslant 1,
\end{align*}
where $\eta_k \geqslant 0$ are constants with $\sum_{k=1}^{\infty}\eta_k <\infty$, then $\{h_k\}_{k\geq1}$ converges almost surely with $\sum_{k=1}^{\infty} v_k < \infty$.
\end{lemma}

\subsection{Problem Formulation}\label{sec2.2}
Consider the following distributed convex optimization problem over a network with $n$ nodes/agents:
\begin{align}\label{2-1}
\min \;\;&F(\xi)=\sum_{i=1}^{n} f^{i}(\xi)\notag\\
s.\;t.\; &\xi\in X=\bigcap_{i=1}^{n}X_{i},
\end{align}
where $f^{i}(\cdot)$ is a local non-smooth convex objective function corresponding to agent $i$ and $X_{i}\subset\mathcal{R}^{m}$ is the local bounded closed convex constraint set known by agent $i$ only. Without loss of generality, we assume that the set $X$, i.e., the intersection of $X_{i},~i=1,\cdots,n$, is non-empty.

The communication topology among agents is modeled by time-varying networks $\mathcal{G}_{k}=(\mathcal{N},
\mathcal{E}_{k}, W_{k}),~k\geq1$,  where $\mathcal{N} = \{1,2,...,n\}$ is the agent set, $k$ is the time index,  $\mathcal{E}_{k}\subset
\mathcal{N}\times \mathcal{N} $ is the edge set at time $k$ which represents the information communication among agents and $W_{k}=\big[w^{ij}_{k}\big]_{ij}$ is the adjacency matrix of $\mathcal{E}_{k}$. The term $w^{ij}_{k}$ denotes the $ij$-th entry of matrix $W_{k}$.   The neighbor set of agent $i$ at time $k$ is represented by $\mathcal{N}^i_k$, i.e., $\mathcal{N}^i_k=\{j\in\mathcal{N}~|~(j,i)\in\mathcal{E}_k\}$. The observation of agent $i$ at time $k$ is its own function value $f^i(\xi_k)$ and those of its neighbors', i.e., $f^j(\xi_k),~j\in\mathcal{N}^i_k$. In the random environment, both $f^i(\xi_k)$ and $f^j(\xi_k),~j\in\mathcal{N}^i_k$ are corrupted by noises, i.e., the observations of agent $i$ at time $k$ being $y_{k+1}^i=f^i(\xi_k)+\epsilon^i_{k}$ and $y_{k+1}^j=f^j(\xi_k)+\epsilon^j_{k},~j\in\mathcal{N}^i_k$.

In addition to observations $y_k^i$, the majority of existing distributed stochastic optimization algorithms for solving problem (\ref{2-1}) demand the measurements of subgradients of the local non-smooth objective functions. However, in practice, the subgradient information is not always available and as pointed in \cite{nesterov2005lexicographic}, the computational complexity of the subgradient is related to the dimension of the function arguments, which may be very high in the network environment.

In this paper, we propose a class of subgradient-free algorithms for solving problem (\ref{2-1}). The first step is to construct the searching direction for each local objective function. This is achieved by introducing stochastic dithers to each agent and then constructing randomized difference as an estimate for the subgradient. The details are as follows.

For each agent $i\in \mathcal{N}$, we introduce a sequence of dither signals $\bigtriangleup_{k}^{i}\in \mathcal{R}^m$, $k\geqslant 0$ with $\bigtriangleup_{k}^{i} \triangleq \big[\bigtriangleup_{k}^{i1},\bigtriangleup_{k}^{i2},\ldots,\bigtriangleup_{k}^{im}\big]^{\top}$. Let $\{c_k\}_{k\geq0}$ be a sequence of positive constants decreasing to zero. After adding dithers into each agent, the corresponding observations are $\big[y_{k+1}^{i}\big]^{+}=f^{i}\big(x^i_{k}+c_{k}\bigtriangleup_{k}^{i}\big)+\big[\epsilon^{i}_{k}\big]^{+}$ and $\big[y_{k+1}^{i}\big]^{-}=f^{i}\big(x^i_{k}-c_{k}\bigtriangleup_{k}^{i}\big)+\big[\epsilon^{i}_{k}\big]^{-}$, where $\big[\epsilon^{i}_{k}\big]^{+}$ and $\big[\epsilon^{i}_{k}\big]^{-}$ are the corresponding observation errors and by $[\cdot]^+$ and $[\cdot]^-$ it means the dithers are in the positive and negative directions.

Define $\epsilon_{k}^{i}\triangleq[\epsilon^{i}_{k}]^{+}-[\epsilon^{i}_{k}]^{-}$. Three types of randomized differences are given as follows:
\begin{itemize}
\item[(a)]\textbf{Right-sided randomized differences:}
\begin{align}\label{2-2}
\big[d^i_{k}\big]^{+}=\dfrac{\Big[\big[y_{k+1}^{i}\big]^{+}-y_{k+1}^{i}\big)\Big]\big[\bigtriangleup_{k}^{i}
\big]^{-}}{c_{k}};
\end{align}
\item[(b)]\textbf{Left-sided randomized differences:}
\begin{align}\label{2-3}
\big[d^i_{k}\big]^{-}=\dfrac{\Big[y_{k+1}^{i}-\big[ y_{k+1}^{i}\big]^{-}\Big]\big[\bigtriangleup_{k}^{i}
\big]^{-}}{c_{k}};
\end{align}
\item[(c)]\textbf{Two-sided randomized differences:}
\begin{align}\label{2-4}
d^i_{k}=\dfrac{\Big[\big[y_{k+1}^{i}\big]^{+}-\big[y_{k+1}^{i}\big]^{-}\Big]\big[\bigtriangleup_{k}^{i}
\big]^{-}}{2c_{k}},
\end{align}
\end{itemize}
where $\big[\bigtriangleup_{k}^{i}\big]^{-}\triangleq\Big[\frac{1}{\bigtriangleup_{k}^{i1}},\frac{1}{\bigtriangleup_{k}^{i2}},\ldots,\frac{1}{\bigtriangleup_{k}^{im}}\Big]^{\top}$.

\begin{remark}\label{Rem1}
The randomized differences defined as above serve as the searching direction of each objective function as well as the estimates for the subgradient. By Definition \ref{Def1}, the randomized differences do not fall into the subgradient category. In the following, we will give the randomized differences based distributed stochastic optimization algorithms and establish their asymptotic properties.
\end{remark}

\section{Distributed Algorithm and Hypotheses}\label{sec3}

In this section, we first propose three distributed subgradient-free algorithms with randomized differences. Then we introduce system hypotheses and some technical lemmas to be used for convergence analysis of these proposed algorithms.

\subsection{Distributed Subgradient-Free Algorithm with Randomized Differences}\label{sec3.1}

For agent $i~(i\in\mathcal{N})$, the design of our distributed subgradient-free algorithms with two-sided randomized differences is given as follows. Let  $\xi^{i}_{k}$  symbolize the state of agent $i$ at time $k$. For the $(k+1)$-th iteration, agent $i$ first collects the states from its active neighbors, i.e., $\xi^{j}_{k},~j\in\mathcal{N}^i_k$ and computes a local average of $\xi^{j}_{k},~j\in\mathcal{N}^i_k$ to update its own state, denoted by $x^i_k$. Then agent $i$ calculates a randomized difference  based on its local function $f^{i}(x)$ at $x=x^i_{k}+c_{k}\bigtriangleup_{k}^{i}$,  $x=x^i_{k}-c_{k}\bigtriangleup_{k}^{i}$. After this, the $(k+1)$-th estimates $\xi^i_{k+1}$ is obtained from a iterative descent algorithm and a projection operator.  The algorithms with one-sided randomized differences can be designed similarly with $d^{i}_{k}$ replaced by $\big[d^i_{k}\big]^{+}$ and $\big[d^i_{k}\big]^{-}$, respectively.  In fact, distributed subgradient-free algorithms with both two-sided and one-sided randomized differences, are given in Algorithm \ref{Alg1}:
\begin{algorithm}[h]
	\flushleft
	\caption{ \bf Distributed Subgradient-Free Algorithms with Randomized Differences}\label{Alg1}
	\begin{algorithmic}[1]
		\State Initialization of $\xi^{i}_{1}\in X_{i}$ for all $i=1,2,\ldots n$. Choose stepsize sequence $\{\iota_{k}\}$ and dither sequence $\{c_{k}\}$.
		\State Average of local observations:
		\begin{align}\label{3-1}
		x^{i}_{k}=\sum_{j\in\mathcal{N}_i}w^{ij}_k\xi^{j}_k=\sum_{j=1}^{n}w^{ij}_k\xi^{j}_k.
		\end{align}
		\State Calculation of randomized difference $d^{i}_{k}$, or $\big[d^i_{k}\big]^{+}$, or $\big[d^i_{k}\big]^{-}$:
		\begin{subequations}
		\begin{align}	&d^i_{k}=\dfrac{\Big[\big[y_{k+1}^{i}\big]^{+}-\big[y_{k+1}^{i}\big]^{-}\Big]\big[\bigtriangleup_{k}^{i}
			\big]^{-}}{2c_{k}}~~~\mathrm{Two-sided};\label{3-2-1}\\			&\big[d^i_{k}\big]^{+}=\dfrac{\Big[\big[y_{k+1}^{i}\big]^{+}-y_{k+1}^{i}\big)\Big]\big[\bigtriangleup_{k}^{i}
\big]^{-}}{c_{k}}~~\mathrm{Right-sided};\label{3-2-2}\\
			&\big[d^i_{k}\big]^{-}=\dfrac{\Big[y_{k+1}^{i}-\big[ y_{k+1}^{i}\big]^{-}\Big]\big[\bigtriangleup_{k}^{i}
\big]^{-}}{c_{k}}~~~\mathrm{Left-sided}\label{3-2-3}.
		\end{align}
		\end{subequations}
		\State Descent step:
		\begin{subequations}
		\begin{align}	\hat{\xi}^{i}_{k+1}&=x^{i}_k-\iota_{k}d^i_{k}~~~\mathrm{Two-sided};\label{3-3-1}\\	\hat{\xi}^{i}_{k+1}&=x^{i}_k-\iota_{k}\big[d^i_{k}\big]^{+}~~~\mathrm{Right-sided}; \label{3-3-2}\\		\hat{\xi}^{i}_{k+1}&=x^{i}_k-\iota_{k}\big[d^i_{k}\big]^{-}~~~\mathrm{Left-sided}\label{3-3-3}.
		\end{align}
		\end{subequations}
		\State Projection step:\begin{align}\label{3-4}
		\xi^{i}_{k+1}=P_{X_{i}}\big(\hat{\xi}^{i}_{k+1}\big).
		\end{align}
		\State{Check the end condition of algorithm. If the condition is satisfied, then the algorithm is terminated. Otherwise, $k:=k+1$ and go to Step 2.}
	\end{algorithmic}
\end{algorithm}

The reference \cite{yuan2015randomized} only considered the right-sided randomized differences of Gaussian approximation functions and we consider three different types of randomized differences in designs of our proposed algorithms. It is also worth noting that our proposed algorithm is also different from the distributed  KW algorithm in \cite{anit2018distributed}, whose construction is based on deterministic differences with periodic dithers for smooth strongly-convex objective functions.

In the following, we will introduce conditions to guarantee the global minimization of the proposed algorithms.

\subsection{Hypotheses and Technical Lemmas}\label{sec3.2}

We first introduce a hypothesis on local objective functions $f^{i}(\cdot),~i=1,\cdots,n$:
\begin{hypothesis}\label{Hyp1}
\
\begin{itemize}
\item[(a)] $f^{i}(\cdot)$ $i=1,2,\ldots,n$ are convex but non-smooth functions with subgradients, denoted by $\partial f^i(\cdot)$.
\item[(b)] There exists a positive constant $L$ such that for any $x\in \mathrm{dom}(f^{i})$, $\big\|\partial f^{i}(x)\big\|_{2}\leqslant L$, $i=1,2,\ldots,n$.
\end{itemize}
\end{hypothesis}

\begin{remark}\label{Rem2}
According to Definition \ref{Def1} and Hypothesis \ref{Hyp1} (b), for any $x_{i}, y_{i}\in X_{i}$, where $X_{i}$ is a bounded closed convex set, $f^{i}$ is Lipschitz over $X_{i}$
\begin{align*}
\big\|f^{i}(x)-f^{i}(y)\big\|_{2}\leqslant \|\langle\partial  f^{i}(x), x-y\rangle \|_{2}\leqslant  L\big\|x-y\big\|_{2}, \quad i=1,2,\ldots,n.
\end{align*}
\end{remark}

Hypothesis \ref{Hyp1} requires the convexity of  local objective functions and the boundedness of their subgradients. Hypothesis \ref{Hyp1} is a traditional condition used for distributed first-order optimization, c.f., \cite{nedic2009distributed, nedic2010constrained} and distributed zeroth-order algorithms  \cite{hajinezhad2017zeroth}. Note that the reference \cite{anit2018distributed} assumed that local objective functions are twice continuously differentiable and strongly convex, which make the problem much easier.   Here, although the existence of the subgradient is required, it is not involved in our proposed algorithm, which  is, therefore, called subgradient-free.

Next, we introduce a connectivity condition for the time-varying network
$\mathcal{G}_{k}=(\mathcal{N},\mathcal{E}_{k}, W_{k})$.

\begin{hypothesis}\label{Hyp2}
The graph $\mathcal{G}_{k}=(\mathcal{N},
\mathcal{E}_{k}, W_{k})$ satisfies the following conditions:
\begin{itemize}
\item[(a)]There exists a constant $\eta$ with $0<\eta<1$ such that $\forall k\geqslant 0$ and $\forall i, j\in\mathcal{N}$, $ w^{ii}_{k}\geqslant \eta$ and $w^{ij}_{k}\geqslant \eta$ if $(j,i)\in \mathcal{E}_{k}$.
\item[(b)]$W_{k}$ is doubly stochastic, i.e. $\sum_{i=1}^{n}w^{ij}_{k}=1$ and $\sum_{j=1}^{n}w^{ij}_{k}=1$.
\item[(c)]There is an integer $\kappa\geqslant 1$ such that $\forall k\geqslant 0$ and $\forall i,j\in \mathcal{N}$,
\begin{equation*}
 (j,i)\in \mathcal{E}_{k}\cup \mathcal{E}_{k+1}\cup \cdots \cup \mathcal{E}_{k+\kappa-1}.
\end{equation*}
\end{itemize}
\end{hypothesis}
Hypothesis \ref{Hyp2} is widely applied in the literature of distributed (stochastic) optimization for time-varying networks (see., e.g.,  \cite{nedic2009distributed,nedic2010constrained}). It indicates that each agent $i$ can gather information from all its neighbors at least once during each period of $\kappa$, though the network can be disconnected at each time $k$ and is time-varying.

We present the following two hypotheses  for the parameter selection of the proposed algorithm.

Define a sequence of $\sigma$-algebras $ F_{k}\triangleq\sigma\big\{x^{i}_{k},x^{i}_{k-1},\cdots,x^{i}_{0},i=1,\cdots,n;~\epsilon^{i}_{k-1},\epsilon^{i}_{k-2},\cdots,\\\epsilon^{i}_{0}, i=1,\cdots,n;~\bigtriangleup^{i}_{k-1},\bigtriangleup^{i}_{k-2},\cdots,\bigtriangleup^{i}_{0},i=1,\cdots,n\big\}$,  $k\geqslant 1$.  We further make the following hypothesis on the dither signal $\bigtriangleup_{k}^{i}$ and the observation noise $\epsilon_{k}^{i}$:

\begin{hypothesis}\label{Hyp3}
\
\begin{itemize}
\item[(a)]For any fixed $i\in\{1,\cdots,n\}$ and $p\in\{1,\cdots,m\}$,  $\big\{\bigtriangleup_{k}^{ip}\big\}_{k\geqslant 0}$ is chosen as a sequence of independent and identically distributed (i.i.d.) random variables such that
\begin{align*}
\big|\bigtriangleup_{k}^{ip}\big|<a,\; \;\bigg|\frac{1}{\bigtriangleup_{k}^{ip}}\bigg|<b,\;\;\mathbb{E}\bigg[\frac{1}{\bigtriangleup_{k}^{ip}}\bigg]=0, \quad k\geqslant 0,\quad \forall (i,p).
\end{align*}
\item[(b)]For $i\neq j$ or $p\neq q$, the sequences $\big\{\bigtriangleup_{k}^{ip}\big\}_{k\geqslant 0}$ and $\big\{\bigtriangleup_{k}^{jq}\big\}_{k\geqslant 0}$ are mutually independent.
\item[(c)]For any fixed $i,j\in\{1,\cdots,n\}$ and $p\in\{1,\cdots,m\}$, the dither signal $\big\{\bigtriangleup_{k}^{ip}\big\}_{k\geqslant 0}$ and the noise $\big\{\epsilon_{k}^{j}\big\}_{k\geqslant 0}$ are  mutually independent.
\item[(d)]For any $k\geqslant 0$, $ \mathbb{E}\big[\epsilon_{k+1}^{i} | F_{k}\big]=0$ and $\sup_{k\geq0,i=1,\cdots,n} \mathbb{E}\big\|\epsilon_{k}^{i}\big\|_{2}^{2}<\infty$.
\end{itemize}
\end{hypothesis}
Then, we introduce conditions on the step-size $\iota_{k}$  of  distributed subgradient-free algorithms with randomized differences (Algorithm \ref{Alg1}) and the coefficient $c_{k}$ used in the randomized differences (\ref{2-2})-(\ref{2-4}):

\begin{hypothesis}\label{Hyp4} Both $\{\iota_k\}_{k\geq1}$ $\{c_k\}_{k\geq1}$ are positive sequences tending to zero such that
\begin{itemize}
\item[(a)]$\iota_{k}>0$, $\sum_{k=1}^{\infty}\iota_{k}<\infty$.
\item[(b)]$c_{k}>0$, $c_{k}\rightarrow 0$.
\item[(c)] $\sum_{k=1}^{\infty}\frac{\iota_{k}}{c_{k}}=\infty$, $\sum_{k=1}^{\infty}\frac{\iota_{k}^{2}}{c_{k}^{2}}<\infty$, and $\sum_{k=1}^{\infty}{\iota_{k}}{c_{k}}<\infty$.
\end{itemize}
\end{hypothesis}
For any $k\geqslant s$, set the transition matrix of $W_{k}$ as $\Psi_{k,s}\triangleq W_{k}W_{k-1}\cdots W_{s}$. Denote $\big[\Psi_{k,s}\big]_{ij}$ as the $(i,j)$-th entry of  $\Psi_{k,s}$.   The following lemma given in \cite{ram2010distributed, nedic2009distributed}, describes the proposition of transition matrix  $\big[\Psi_{k,s}\big]_{ij}$ of the considered time-varying network.

\begin{lemma}\label{Lem4}\cite{nedic2009distributed, ram2010distributed}
 If Hypothesis \ref{Hyp2} holds, then $\Big|\big[\Psi_{k,s}\big]_{ij}-\frac{1}{n}\Big|\leqslant \lambda \beta^{k-s},\;\forall
		k>s$, where $\lambda=2\big(1+\eta^{-K_{0}}\big)/\big(1-\eta^{-K_{0}}\big)$
with $K_{0}=\big(n-1\big)\kappa$ and $\beta=\big(1-\eta^{-K_{0}}\big)^{1/K_{0}}<1$.
\end{lemma}

\section{Main Results}\label{sec4}

In this section, we introduce three main theorems step by step.   At first, we establish
the consensus with probability one of distributed subgradient-free algorithm with two-sided randomized differences. Then we prove that this algorithm achieves  the global minimizer with probability one.  Finally, we also show the mean-square  convergence rate of estimates obtained from Algorithm \ref{Alg1}.

We will mainly focus on the analysis of the distributed subgradient-free algorithm with two-sided randomized  differences in Algorithm \ref{Alg1}. The analysis of distributed subgradient-free algorithm with one-sided randomized  differences can be given similarly and thus is omitted.

First, we introduce a theorem regarding consensus analysis of the proposed algorithm.

\begin{theorem}\label{The1}
Under Hypotheses \ref{Hyp1}-\ref{Hyp4}, the consensus among estimates
$\{ \xi^i_{k}\},\;i\in\mathcal{N}$ generated from distributed subgradient-free algorithm with two-sided randomized differences in Algorithm \ref{Alg1} is achieved almost surely (a.s.).
\end{theorem}

To prove Theorem \ref{The1}, we need the following technical lemmas.

\begin{lemma}\label{Lem5}
If Hypotheses \ref{Hyp1} and \ref{Hyp3} hold, then the first and second order moments of randomized difference $d_{k}^{i}$ are bounded by
\begin{align*}
\mathbb{E}\big\|d_{k}^{i}\big\|_{2}\leqslant L+\dfrac{\sqrt{m}be}{2c_{k}},
\end{align*}
and
\begin{align*}
\mathbb{E}\big\|d_{k}^{i}\big\|_{2}^{2}\leqslant \bigg(L+\dfrac{\sqrt{m}be}{2c_{k}}\bigg)^{2}
\end{align*}
respectively, where $L$ is the Lipschitz constant given in Remark \ref{Rem2}, $b$ is a constant given in  Hypothesis \ref{Hyp3} (a), and $e=\Big(\sup\limits_{k\geq0,i=1,\cdots,n} \mathbb{E}\big\|\epsilon_{k}^{i}\big\|_{2}^{2}\Big)^{\frac{1}{2}}$.
\end{lemma}

\begin{proof}
See Appendix.
\end{proof}

Define $\bar{\xi}_{k+1}\triangleq\frac{1}{n}\sum\limits_{i=1}^n\xi^i_{k+1}$ as the average of states $\xi^i_{k+1}$'s.
\begin{lemma}\label{Lem6}
 If Hypotheses \ref{Hyp1}-\ref{Hyp4} hold, then for any $ i\in\mathcal{N}$,
$$
\lim_{k \rightarrow \infty}\mathbb{E}\big\|\xi^{i}_{k}-\bar{\xi}_{k}\big \|_{2}=0,
$$
i.e., the consensus being achieved in $L_{1}$ for estimates generated from the distributed subgradient-free algorithm \ref{Alg1}.
\end{lemma}
\begin{proof}
See Appendix.
\end{proof}
\begin{lemma}\label{Lem7}
With Hypotheses \ref{Hyp1} and \ref{Hyp3}, we have
\begin{align*}
 \sum_{k=1}^{\infty}\frac{\iota_{k}}{c_k}\big\|\xi^i_{k}-\bar{\xi}_{k}\big\|_{2}<\infty~~\mathrm{a.s.},\quad i=1,2,\ldots,n.
\end{align*}
\end{lemma}

\begin{proof}
The proof of lemma \ref{Lem7} can follow the proof of Theorem 6.1 in  \cite{ram2010distributed} by replacing the subgradients with the randomized differences and noticing the boundedness of $\{\bigtriangleup_k\}_{k\geq0}$ in Hypothesis \ref{Hyp3}, the Lipschitz of $f_i(\cdot)$ guaranteed by Remark \ref{Rem2} and the step size condition in Hypothesis \ref{Hyp4}. Thus the detailed proof is omitted.	
\end{proof}

Then it is time to prove Theorem \ref{The1}.

\begin{proof}[Proof of Theorem \ref{The1}]
 According to Lemma \ref{Lem6},
$\lim_{k\rightarrow\infty}\mathbb{E}\big\|\xi^i_{k+1}-\bar{\xi}_{k+1}\big\|_{2}=0$ holds. Then, by Fatou's Lemma (Theorem 1.5.4 \cite{durrett2010probability}), the following inequality takes place
\begin{align}\label{4-2}
0\leqslant\mathbb{E}\Big[\underset{k\rightarrow\infty}{\liminf}\big\|\xi^i_{k+1}-\bar{\xi}_{k+1}\big\|_{2}\Big]\leqslant\underset{k\rightarrow\infty}{\liminf}\mathbb{E}\big\|\xi^i_{k+1}
-\bar{\xi}_{k+1}\big\|_{2}=0,
\end{align}
which yields $\mathbb{E}\Big[\underset{k\rightarrow\infty}{\liminf}\big\|\xi^i_{k+1}-\bar{\xi}_{k+1}\big\|_{2}\Big]=0$. Therefore,
\begin{align}\label{4-3}
\underset{k\rightarrow\infty}{\liminf} \big\|\xi^i_{k+1}-\bar{\xi}_{k+1}\big\|_{2}=0.
\end{align}
By Lemma \ref{Prop1}, we have
\begin{align}\label{4-4}
\sum_{i=1}^{n}\big\|\xi^i_{k+1}-\bar{\xi}_{k+1}\big\|_{2}^{2}\leqslant \sum_{i=1}^{n}\big\|\xi^{i}_{k+1}-\bar{\xi}_{k}\big\|_{2}^{2}.
\end{align}

According to Lemma \ref{Lem2}(b), $\big\|\xi^i_{k+1}-\bar{\xi}_{k}\big\|_{2}^2\leqslant\big\|\hat{\xi}^{i}_{k+1}-\bar{\xi}_{k}\big\|_{2}^2$. Based on the above two inequalities and (\ref{3-3-1}), we have
\begin{align}\label{4-5}
\nonumber &\sum_{i=1}^{n}  \big \|\xi^i_{k+1}-\bar{\xi}_{k+1}\big\|_{2}^{2}\leqslant \sum_{i=1}^{n} \big \|\hat{\xi}^{i}_{k+1}-\bar{\xi}_{k}\big\|_{2}^2
\notag\\\leqslant  &\sum_{i=1}^{n} \sum_{j=1}^{n}w^{ij}_{k}\big\|\xi^j_{k}-\bar{\xi}_{k}\big\|_{2}^2+\iota^2_k\sum_{i=1}^{n}\big\|d^{i}_{k}\big\|_{2}^2+2\iota_{k} \sum_{i=1}^{n}\big \|d^{i}_{k}\big\|_{2}\sum_{j=1}^{n}w^{ij}_{k}\big\|\xi^j_{k}-\bar{\xi}_{k}\big\|_{2}.
\end{align}

From Hypothesis \ref{Hyp2}(b), $\sum_{i=1}^{n}\sum_{j=1}^{n}w^{ij}_{k}\big\|\xi^{j}_{k}-\bar{\xi}_{k}\big\|_{2}^{2}=\sum_{i=1}^{n}\big\|\xi^{i}_{k}-\bar{\xi}_{k}\big\|_{2}^{2}$. Thus, by taking the conditional expectation to both sides of \eqref{4-5} and noticing that $\xi_k^j$ is $F_k$-measurable, we have
\begin{align}\label{4-6}
\sum_{i=1}^{n}\mathbb{E}\Big[\big \|\xi^i_{k+1}-\bar{\xi}_{k+1}\big\|_{2}^{2}\big|F_k\Big]
\leqslant  \sum_{i=1}^{n}\big \|\xi^i_{k}-\bar{\xi}_{k}\big\|_{2}^2+ \sum_{i=1}^{n}A_{k}^{i}(1)+\sum_{i=1}^{n}A_{k}^{i}(2),
\end{align}
where $A_{k}^{i}(1)=\iota^2_k \mathbb{E}\left[\big\|d^{i}_{k}\big\|_{2}^2\big|F_k\right]$ and $A_{k}^{i}(2)=2\iota_{k}  \mathbb{E}\left[\big\|d^{i}_{k}\big\|_{2}\big| F_k\right]\sum_{j=1}^{n}\big\|\xi^j_{k}-\bar{\xi}_{k}\big\|_{2}$.

Due to Hypothesis \ref{Hyp4} and Lemma \ref{Lem5}, $\sum_{k=1}^{\infty}A_{k}^{i}(1)<\infty$ a.s.  From Lemma \ref{Lem7}, $\sum_{k=1}^{\infty}\frac{\iota_{k}}{c_k}\big\|\xi^j_{k}-\bar{\xi}_{k}\big\|_{2}<\infty$ with probability $1$ and then $\sum_{k=1}^{\infty}A_{k}^{i}(2)<\infty$. Therefore, $\lim_{k\rightarrow\infty}\big\|\xi^i_{k}-\bar{\xi}_{k}\big\|_{2}^2=0$ a.s. by Lemma \ref{Lem3}.
\end{proof}

Then we have the following convergence results.

\begin{theorem}\label{The2}
Set $\iota_{k}=\frac{1}{k^{1+\epsilon}}$ and $c_{k}=\frac{1}{k^{\delta}}$ with $ \frac{1}{2}+\epsilon>\delta\geqslant \epsilon>0$. Under Hypotheses \ref{Hyp1}-\ref{Hyp4}, all the
sequences $\{ \xi^i_{k}\},\;i\in\mathcal{N}$ generated from distributed subgradient-free algorithm with two-sided randomized differences in Algorithm \ref{Alg1} converge to the optimal solution $\xi^{*}$ almost surely.
\end{theorem}
To prove Theorem \ref{The2}, we need the following lemmas, whose proofs are also given in Appendix.

\begin{lemma}\label{Lem8}
If Hypotheses \ref{Hyp1} and \ref{Hyp3} hold, then the following inequalities take place
\begin{align*}
&\mathbb{E}\big[\big\langle d^{i}_{k},\;x^{i}_{k}-\xi^*\big\rangle \big|F_{k}\big] \geqslant f^{i}\big(\bar{\xi}_{k}\big)-f^{i}\big(\xi^*\big)-L\big\|x^{i}_{k}-\bar{\xi}_{k}\big\|_{2}-2c_{k}L\mathbb{E}\big\|\bigtriangleup_{k}^{i}\big\|_{2}-2L,\\
&\mathbb{E}\big[\big\langle d^{i}_{k},\;x^{i}_{k}-\xi^*\big\rangle\big] \geqslant\mathbb{E}\big[ f^{i}\big(\bar{\xi}_{k}\big)-f^{i}\big(\xi^*\big)\big]-L\mathbb{E}\big\|x^{i}_{k}-\bar{\xi}_{k}\big\|_{2}-2c_{k}L\mathbb{E}\big\|\bigtriangleup_{k}^{i}\big\|_{2}-2L,
\end{align*}
where $L$ is an positive constant given in Hypothesis \ref{Hyp1}.
\end{lemma}

\begin{lemma}\label{Lem9}
Under Hypotheses \ref{Hyp1}-\ref{Hyp4} and Lemma \ref{Lem8}, all the
sequences $\{ \xi^i_{k}\},\;i\in\mathcal{N}$ generated from distributed subgradient-free algorithm with two-sided randomized differences in Algorithm \ref{Alg1} converge almost surely (a.s.) to a random variable for Problem \ref{2-1}.
\end{lemma}

\begin{lemma}\label{Lem10}
Set $\iota_{k}=\frac{1}{k^{1+\epsilon}}$ and $c_{k}=\frac{1}{k^{\delta}}$ with $ \frac{1}{2}+\epsilon>\delta\geqslant \epsilon>0$. Under Hypotheses \ref{Hyp1}-\ref{Hyp4} and Lemma \ref{Lem8}, we have
\begin{align*}
\sum_{i=1}^{n}\mathbb{E} \big\|\xi^i_{k+1}-\xi^{*}\big\|_{2}^{2}\leqslant  \frac{M_{1}}{k^{1+2\epsilon-2\delta}}+\dfrac{M_{2}}{k^{\epsilon}}+ \frac{M_{3}}{k^{\epsilon+\delta}}+\dfrac{M_{4}}{k^{\epsilon}}.
\end{align*}
where $M_{1},\ldots, M_{4}$ are constants.
\end{lemma}
\begin{proof}[Proof of Theorem \ref{The2}]
According to Lemma \ref{Lem9}, $\lim_{k\rightarrow \infty }\sum_{i=1}^{n} \big\|\xi^i_{k}-\xi^{*}\big\|_{2}^{2}$ converges to a non-negative random variable almost surely. According to Lemma \ref{Lem10}, we have
\begin{align*}
\sum_{i=1}^{n}\mathbb{E} \big\|\xi^i_{k}-\xi^{*}\big\|_{2}^{2} \sim O\Big(\max\Big\{\frac{1}{k^{\epsilon}}, \frac{1}{k^{1+2\epsilon-2\delta}}\Big\}\Big).
\end{align*}
 which means that $\{ \xi^i_{k}\},\;i\in\mathcal{N}$ generated from the distributed subgradient-free algorithm of two-sided randomized differences in Algorithm \ref{Alg1} converge to the optimal solution $\xi^{*}$ in $L_{2}$. Therefore,
\begin{align}\label{4-24}
\lim_{k\rightarrow \infty }\sum_{i=1}^{n} \big\|\xi^i_{k}-\xi^{*}\big\|_{2}^{2}=0, \quad \mathrm{a.s.}
\end{align}
\end{proof}

\begin{remark}
For centralized KW algorithms, it usually assumes that the objective function is smooth and by the random difference technique, the optimization problem is transformed into the root-searching of the zero of the corresponding gradient function. In this paper, since algorithms are formulated into a distributed manner and constraints are also considered, the problem cannot be formulated into the root-searching of the gradient functions and thus we give the mathematical analysis in a different way. Compared with the analysis given in \cite{nedic2010constrained}, since the stochastic dither is introduced into the objective function, i.e., $f(x_k^{i}\pm c_{k}\bigtriangleup_{k}^{i})$ with $\bigtriangleup_{k}^{i}$ being independent of the estimate $x_k^{i}$, this makes the proof much more complicated and many efforts are devoted to the analysis of difference term $d^{i}_{k}$. See, e.g., Lemma \ref{Lem5}, Lemma \ref{Lem8}.
\end{remark}

According to Theorems \ref{The1} and \ref{The2}, all the estimates $\{ \xi^i_{k},i\in\mathcal{N}\}$ converge to the optimal solution $\xi^{*}$ almost surely. These two theorems, in fact, establish the almost sure results of the distributed algorithm with two-sided randomized differences. Similar results also hold for distributed subgradient-free algorithms with one-sided randomized differences, whose proofs are omitted due to space limitations.

\begin{corollary}\label{The3}
Under Hypotheses \ref{Hyp1}-\ref{Hyp4}, all the
sequences $\{ \xi^i_{k}\},\;i\in\mathcal{N}$ generated from the distributed subgradient-free algorithms with one-sided randomized differences converge to the optimal solution $\xi^{*}$ almost surely.
\end{corollary}


Here we briefly compare the results in this paper with those in \cite{hajinezhad2017zeroth, yuan2015randomized, yuan2015zeroth}.  For smooth local objective functions, the mean average convergence established in this paper matches the corresponding results in \cite{hajinezhad2017zeroth, yuan2015randomized, yuan2015zeroth}. For non-smooth local objective functions, we establish the almost sure convergence to global minimization while in \cite{hajinezhad2017zeroth, yuan2015randomized, yuan2015zeroth} it establishes the upper bound of the estimation error in mean average sense but the convergence rate is not analyzed. In the following, we will present the convergence rate in mean square sense of the distributed subgradient-free algorithm with two-sided randomized differences, for which the proof can follow from Lemma \ref{Lem10} directly.


\begin{theorem}\label{The4}
Set $\iota_{k}=\frac{1}{k^{1+\epsilon}}$ and $c_{k}=\frac{1}{k^{\delta}}$ with $ \frac{1}{2}+\epsilon>\delta\geqslant \epsilon>0$. Under Hypotheses \ref{Hyp1}-\ref{Hyp4} and Lemma \ref{Lem10}, for distributed subgradient-free algorithm with two-sided randomized differences in Algorithm \ref{Alg1}, we have
\begin{align*}
\sum_{i=1}^{n}\mathbb{E} \big\|\xi^i_{k}-\xi^{*}\big\|_{2}^{2} \sim O\Big(\max\Big\{\frac{1}{k^{\epsilon}}, \frac{1}{k^{1+2\epsilon-2\delta}}\Big\}\Big).
\end{align*}
\end{theorem}
\begin{remark}
It directly follows from Theorem \ref{The4} that the optimal values for $\epsilon$ and $\delta$ are $\epsilon=\dfrac{1}{2}$ and $ \delta=\dfrac{1}{2} $, respectively, which in turn indicate that $\iota_{k}=\frac{1}{k^{\frac{3}{2}}}$, $c_{k}=\frac{1}{k^{\frac{1}{2}}}$, and
\begin{align*}
\sum_{i=1}^{n}\mathbb{E} \big\|\xi^i_{k}-\xi^{*}\big\|_{2}^{2} \sim O\Big(\frac{1}{\sqrt{k}}\Big).
\end{align*}
The rate matches not only the best rate for centralized stochastic approximation algorithms, see \cite{chen2006stochastic} and references therein, but also the best rate given in distributed first-order stochastic subgradient algorithms with a diminishing step-size.
\end{remark}

\section{Simulations}\label{sec5}
In this section, we give a numerical example for further illustration of our proposed algorithms.

Consider a network system with 5 agents. The time-varying communication topology between the agents can be described by Fig. 1.
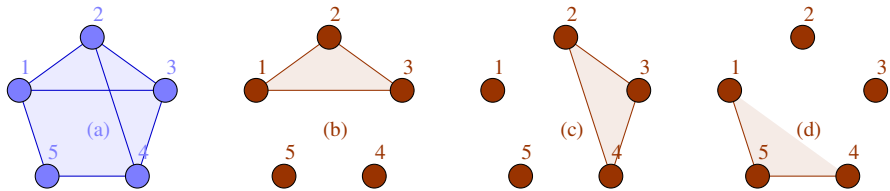
\begin{figure}[!htbp]\label{fig1}
\begin{tikzpicture}[line cap=round,line join=round,>=triangle 45,x=0.6cm,y=0.6cm]
\clip(-3,-3) rectangle (2,2);
\fill[line width=0.8pt,color=ttttff,fill=ttttff,fill opacity=0.10] (-1.,-2.) -- (1.,-2.) -- (1.618,-0.097) -- (0.,1.08) -- (-1.618,-0.097) -- cycle;
\draw [line width=0.4pt,color=qqqqcc] (-1.,-2.)-- (1.,-2.);
\draw [line width=0.4pt,color=qqqqcc] (1.,-2.)-- (1.618,-0.097);
\draw [line width=0.4pt,color=qqqqcc] (1.618,-0.09)-- (0.,1.08);
\draw [line width=0.4pt,color=qqqqcc] (0.,1.08)-- (-1.618,-0.097);
\draw [line width=0.4pt,color=qqqqcc] (-1.618,-0.097)-- (-1.,-2.);
\draw [line width=0.4pt,color=qqqqcc] (-1.618,-0.097)-- (1.618,-0.097);
\draw [line width=0.4pt,color=qqqqcc] (0.,1.08)-- (1.,-2.);
\begin{scriptsize}
\draw [fill=xdxdff] (-1.,-2.) circle (4.5pt);
\draw[color=xdxdff] (-0.86,-1.47) node {5};
\draw [fill=xdxdff] (1.,-2.) circle (4.5pt);
\draw[color=xdxdff] (1.14,-1.47) node {4};
\draw [fill=xdxdff] (1.618,-0.097) circle (4.5pt);
\draw[color=xdxdff] (0.14, -1) node {(a)};
\draw[color=xdxdff] (1.76,0.43) node {3};
\draw [fill=xdxdff] (0.,1.08) circle (4.5pt);
\draw[color=xdxdff] (0.14,1.61) node {2};
\draw [fill=xdxdff] (-1.618,-0.097) circle (4.5pt);
\draw[color=xdxdff] (-1.48,0.43) node {1};
\end{scriptsize}
\end{tikzpicture}
\begin{tikzpicture}[line cap=round,line join=round,>=triangle 45,x=0.6cm,y=0.6cm]
\clip(-3,-3) rectangle (2,2);
\fill[line width=0.8pt,color=zzttqq,fill=zzttqq,fill opacity=0.10] (1.618,-0.097) -- (0.,1.08) -- (-1.618,-0.097) -- cycle;
\draw [line width=0.4pt,color=zzttqq] (1.618,-0.09)-- (0.,1.08);
\draw [line width=0.4pt,color=zzttqq] (0.,1.08)-- (-1.618,-0.097);
\draw [line width=0.4pt,color=zzttqq] (-1.618,-0.097)-- (1.618,-0.097);
\begin{scriptsize}
\draw [fill=zzttqq] (-1.,-2.) circle (4.5pt);
\draw[color=zzttqq] (-0.86,-1.47) node {5};
\draw [fill=zzttqq] (1.,-2.) circle (4.5pt);
\draw[color=zzttqq] (1.14,-1.47) node {4};
\draw [fill=zzttqq] (1.618,-0.097) circle (4.5pt);
\draw[color=zzttqq] (1.76,0.43) node {3};
\draw[color=zzttqq] (0.14, -1) node {(b)};
\draw [fill=zzttqq] (0.,1.08) circle (4.5pt);
\draw[color=zzttqq] (0.14,1.61) node {2};
\draw [fill=zzttqq] (-1.618,-0.097) circle (4.5pt);
\draw[color=zzttqq] (-1.48,0.43) node {1};
\end{scriptsize}
\end{tikzpicture}
\begin{tikzpicture}[line cap=round,line join=round,>=triangle 45,x=0.6cm,y=0.6cm]
\clip(-3,-3) rectangle (2,2);
\fill[line width=0.8pt,color=zzttqq,fill=zzttqq,fill opacity=0.10](1.,-2.) -- (1.618,-0.097) -- (0.,1.08)-- cycle;
\draw [line width=0.4pt,color=zzttqq] (1.,-2.)-- (1.618,-0.097);
\draw [line width=0.4pt,color=zzttqq] (1.618,-0.09)-- (0.,1.08);
\draw [line width=0.4pt,color=zzttqq] (0.,1.08)-- (1.,-2.);
\begin{scriptsize}
\draw [fill=zzttqq] (-1.,-2.) circle (4.5pt);
\draw[color=zzttqq] (-0.86,-1.47) node {5};
\draw [fill=zzttqq] (1.,-2.) circle (4.5pt);
\draw[color=zzttqq] (1.14,-1.47) node {4};
\draw [fill=zzttqq] (1.618,-0.097) circle (4.5pt);
\draw[color=zzttqq] (1.76,0.43) node {3};
\draw[color=zzttqq] (0.14, -1) node {(c)};
\draw [fill=zzttqq] (0.,1.08) circle (4.5pt);
\draw[color=zzttqq] (0.14,1.61) node {2};
\draw [fill=zzttqq] (-1.618,-0.097) circle (4.5pt);
\draw[color=zzttqq] (-1.48,0.43) node {1};
\end{scriptsize}
\end{tikzpicture}
\begin{tikzpicture}[line cap=round,line join=round,>=triangle 45,x=0.6cm,y=0.6cm]
\clip(-3,-3) rectangle (2,2);
\fill[line width=0.8pt,color=zzttqq,fill=zzttqq,fill opacity=0.10] (-1.,-2.) -- (1.,-2.) --(-1.618,-0.097) ;
\draw [line width=0.4pt,color=zzttqq] (-1.,-2.)-- (1.,-2.);
\draw [line width=0.4pt,color=zzttqq] (-1.618,-0.097)-- (-1.,-2.);
\begin{scriptsize}
\draw [fill=zzttqq] (-1.,-2.) circle (4.5pt);
\draw[color=zzttqq] (-0.86,-1.47) node {5};
\draw [fill=zzttqq] (1.,-2.) circle (4.5pt);
\draw[color=zzttqq] (1.14,-1.47) node {4};
\draw [fill=zzttqq] (1.618,-0.097) circle (4.5pt);
\draw[color=zzttqq] (1.76,0.43) node {3};
\draw[color=zzttqq] (0.14, -1) node {(d)};
\draw [fill=zzttqq] (0.,1.08) circle (4.5pt);
\draw[color=zzttqq] (0.14,1.61) node {2};
\draw [fill=zzttqq] (-1.618,-0.097) circle (4.5pt);
\draw[color=zzttqq] (-1.48,0.43) node {1};
\end{scriptsize}
\end{tikzpicture}
\caption{Topology of the $5$-agent network. }
\end{figure}
In Algorithm \ref{Alg1}, the communication topology between the agents is jointly-connected as follows:
the topology is shown in Fig. 1(b) at time $k= 3t$, shown in Fig. 1(c) at time $k= 3t+1$, and shown in Fig. 1(d) at time $k= 3t+1$, for $t=1,2,\ldots$.

Consider the following distributed optimization problem
\begin{align}\label{5-1}
\min \;\;&F(\xi)=\sum_{i=1}^{5} f^{i}(\xi)\notag\\
s.\;t.\; &\xi\in X=\bigcap_{i=1}^{5}X_{i}
\end{align}
with \begin{align}\label{5-2}
\begin{cases}
f^{1}(\xi)=(\xi-3)^{2},X_{1} =\{||\xi||_{2}\leqslant 100\};\\
f^{2}(\xi)=(\xi-2)^{2},X_{2} =\{||\xi||_{2}\leqslant 100\};\\
f^{3}(\xi)=(\xi-1)^{2},X_{3} =\{||\xi||_{2}\leqslant 100\};\\
f^{4}(\xi)=\xi^{2},\quad\quad \;\;X_{4} =\{||\xi||_{2}\leqslant 100\};\\
f^{5}(\xi)=(\xi+1)^{2},X_{5} =\{||\xi||_{2}\leqslant 100\}.
\end{cases}
\end{align}
The distributed optimization problem \eqref{5-1} has a unique minimum at $\xi=1$.
In simulation setup, the sequence of dither signals $\{\bigtriangleup_{k}^{i}\}$ is uniformly distributed on $[-1,-0.5]\cup [0.5,1]$, the noise sequence $\{\epsilon^{i}_{k}\}$ is i.i.d. Gaussian process with distribution $\mathbb{N}(0,1)$. We set the total iteration number to be $T=500$, for $i=1,2,\ldots,5$ and set $\iota_{k}=\frac{1}{k^{\frac{3}{2}}}$ and $c_{k}=\frac{1}{k^{\frac{1}{2}}}$ with  $\frac{1}{2}+\epsilon>\delta\geqslant \epsilon>0$. In the simulation figures, the estimation sequences generated from the subgradient-free algorithm. i. e. Algorithm \ref{Alg1} are denoted by ``SF" and those from the distributed stochastic subgradient projection algorithm in \cite{ram2010distributed} denoted by ``S". In the distributed stochastic subgradient projection algorithm, we set step-size $\frac{1}{k^{\frac{1}{2}}}$ and keep other parameters unchanged. The performances of the algorithms are shown in Fig. \ref{fig2}.
\begin{figure}[htp]
	\centering
	\includegraphics[width=\linewidth]{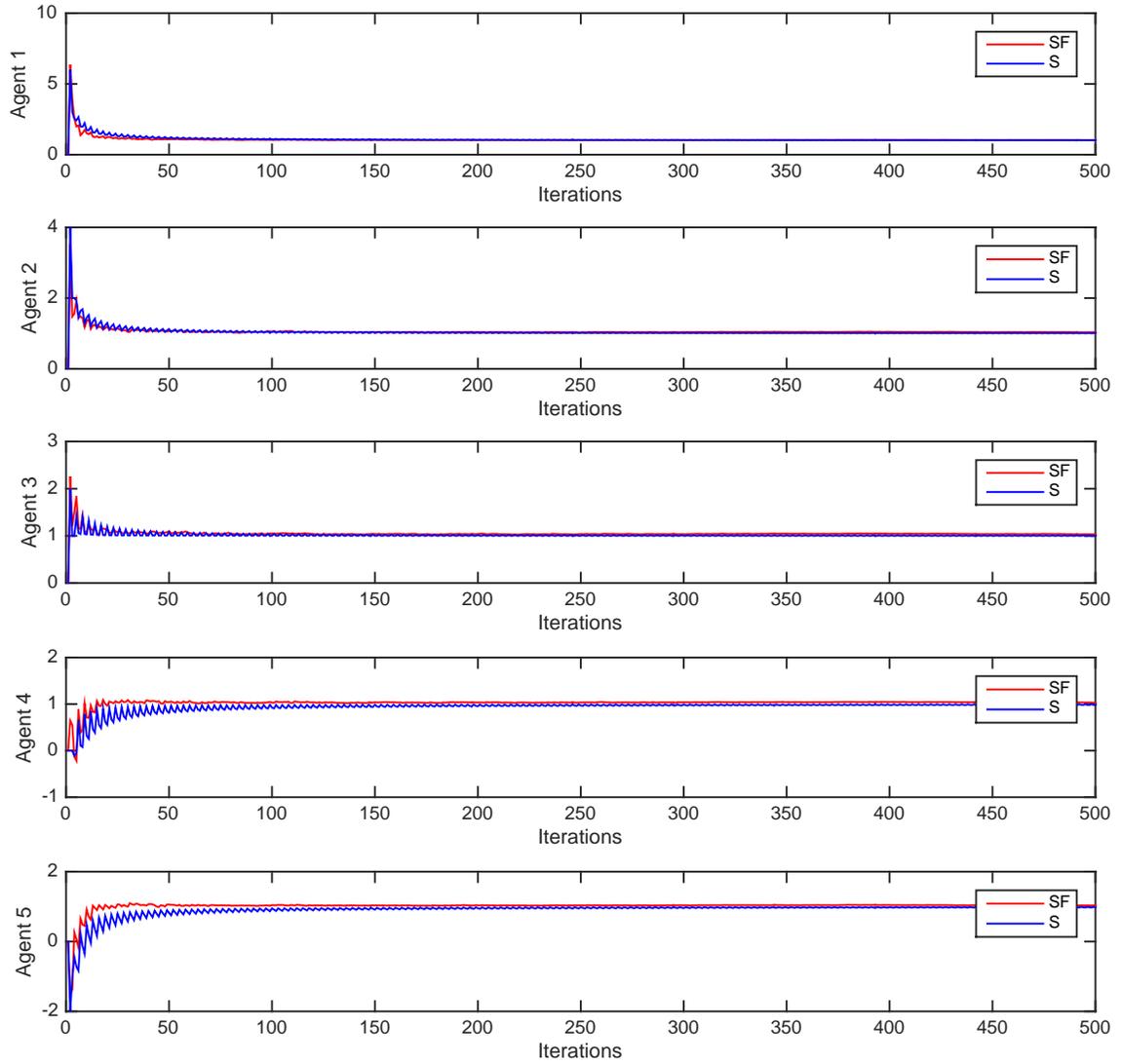}
	\caption{The convergence performance $\xi^{i}_{k}$ of SF(in red lines) and
S \cite{ram2010distributed}  (in blue lines) for each agent $i$}
	\label{fig2}
\end{figure}

Fig. \ref{fig2} shows that $\xi^{i}_{k}$ of each agent $i$ converges to the same optimal point $\xi^{*}=1$ for both Algorithm \ref{Alg1} and the distributed stochastic subgradient projection algorithm in  \cite{ram2010distributed} .

Define
$$\mathrm{R}^{i}(T)= \sum_{i=1}^{n}\mathbb{E} \big\|\xi^i_{T}-\xi^{*}\big\|_{2}^{2}.$$
Figs. \ref{fig3} and \ref{fig4} show the performance of $\mathrm{R}^{i}(T)$'s for Algorithm \ref{Alg1} (SF-Algorithm) and the distributed stochastic subgradient projection algorithm (S-Algorithm) in \cite{ram2010distributed}, respectively.
\begin{figure}[htp]
	\centering
	\includegraphics[width=0.8\linewidth]{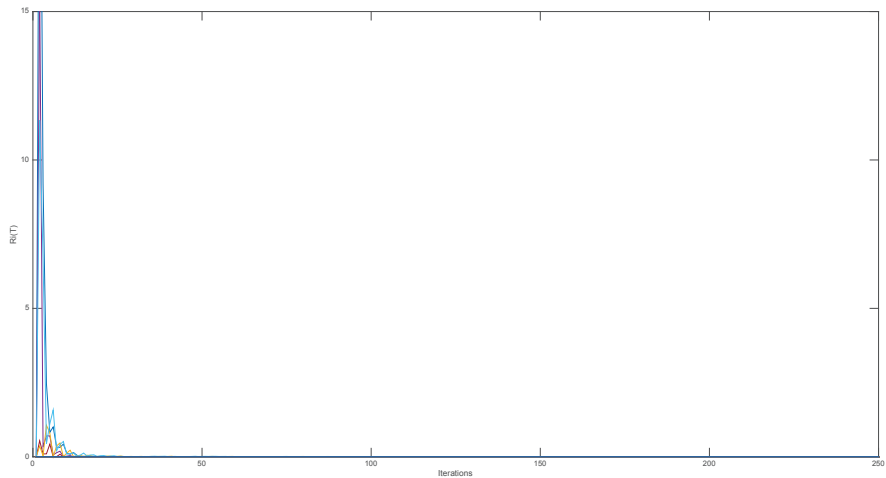}
	\caption{$R^{i}(T)$ for each agent $i$ of SF-Algorithm \ref{Alg1}}.
	\label{fig3}
\end{figure}
\begin{figure}[h]
	\centering
	\includegraphics[width=0.8\linewidth]{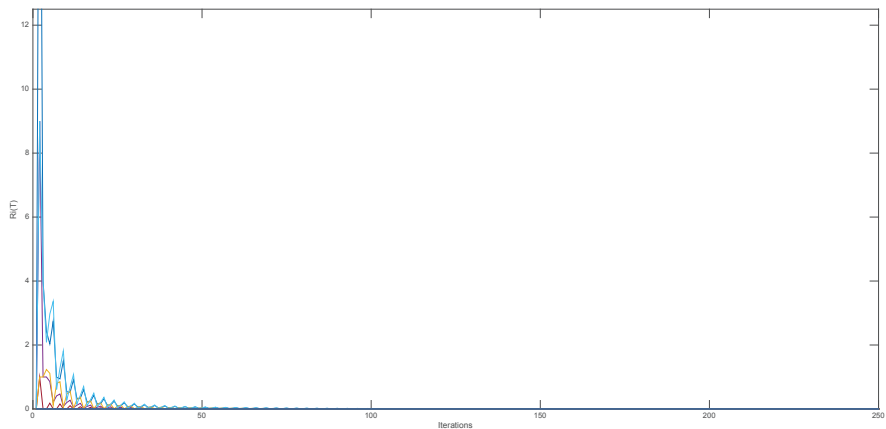}
	\caption{$R^{i}(T)$ for each agent $i$ of S-Algorithm in \cite{ram2010distributed}}.
	\label{fig4}
\end{figure}

From Figs. \ref{fig3} and \ref{fig4} we can see that $\mathrm{R}^{i}(T)$ of each agent $i$ converges to $0$ for both Algorithm \ref{Alg1} and the distributed stochastic subgradient projection algorithm in \cite{ram2010distributed}.

\section{Conclusions}\label{sec6}



In this paper, distributed subgradient-free algorithms with both one-sided and two-sided randomized differences over a time-varying network were introduced with randomized differences technique. The global minimization and the almost sure convergence of estimates were established and the convergence rate in the mean average sense was analyzed as well. The results extended the properties of the subgradient-based distributed stochastic optimization algorithms to the subgradient-free case.


\appendix
Here we give the proofs for Lemmas \ref{Lem5}, \ref{Lem6} and \ref{Lem8}.

\section{Proof of Lemma \ref{Lem5}:}

According to the definition of $d^{i}_{k}$ in \eqref{2-4},
\begin{align}\label{A-1}
d^i_{k}=\dfrac{\Big[f^{i}\big(x^i_{k}+c_{k}\bigtriangleup_{k}^{i}\big)-f^{i}\big(x^i_{k}-c_{k}\bigtriangleup_{k}^{i}\big)\Big]\big[\bigtriangleup_{k}^{i}\big]^{-}}{2c_{k}}+\dfrac{\epsilon_{k}^{i}\big[\bigtriangleup_{k}^{i}\big]^{-}}{2c_{k}}.
\end{align}
Under the Lipschitz property of $f^{i}(\cdot)$ in Remark \ref{Rem2}, we have
$$\big\|f^{i}\big(x^i_{k}+c_{k}\bigtriangleup_{k}^{i}\big)-f^{i}\big(x^i_{k}-c_{k}\bigtriangleup_{k}^{i}\big)\big\|_{2}\leqslant 2Lc_{k}\big\|\bigtriangleup_{k}^{i}\big\|_{2},$$
where $L$ is a positive constant given in Hypothesis \ref{Hyp1}.
\begin{itemize}
\item[(a)]By Hypothesis \ref{Hyp3}(b) and the above inequality, we have $$\mathbb{E}\Bigg\|\dfrac{\Big[f^{i}\big(x^i_{k}+c_{k}\bigtriangleup_{k}^{i}\big)-f^{i}\big(x^i_{k}-c_{k}\bigtriangleup_{k}^{i}\big)\Big]\big[\bigtriangleup_{k}^{i}\big]^{-}}{2c_{k}}\Bigg\|_{2}\leqslant L.$$

Recalling  Hypothesis \ref{Hyp3}(c), $\{\varepsilon_k^i\}$ and $\{\bigtriangleup_{k}^{i}\}$ are mutually independent. By the Lyapunov inequality (see e. g. \cite{chen2006stochastic}), we have the following inequalities,
\begin{align}\label{A-2}
\mathbb{E}\Bigg\|\dfrac{\epsilon_{k}^{i}\big[\bigtriangleup_{k}^{i}\big]^{-}}{2c_{k}}\Bigg\|_{2}&=\frac{1}{2c_{k}}\mathbb{E}\big\|\epsilon^{i}_{k}\big\|_{2}\mathbb{E}\big\|\big[\bigtriangleup_{k}^{i}\big]^{-}\big\|_{2}\notag\\&\leqslant\dfrac{\sqrt{m}b}{2c_{k}}\mathbb{E}\big\|\epsilon^{i}_{k}\big\|_{2}\leqslant \dfrac{\sqrt{m}b}{2c_{k}}\Big(\mathbb{E}\big\|\epsilon_{k}^{i}\big\|_{2}^{2}\Big)^{\frac{1}{2}}\leqslant \dfrac{\sqrt{m}b}{2c_{k}}\Big(\sup_{i} \mathbb{E}\big\|\epsilon_{k}^{i}\big\|_{2}^{2}\Big)^{\frac{1}{2}}.
\end{align}
Therefore, $\mathbb{E}\|d_{k}^{i}\|_{2}\leqslant L+\dfrac{\sqrt{m}b}{2c_{k}}\Big(\sup_{i} \mathbb{E}\big\|\epsilon_{k}^{i}\big\|_{2}^{2}\Big)^{\frac{1}{2}}$ and the first part of Lemma \ref{Lem5} is proved.

\item[(b)] We now establish $\mathbb{E}\big\|d_{k}^{i}\big\|_{2}^{2}\leqslant L$. We first have the following equality:
$$\mathbb{E}\big\|d_{k}^{i}\big\|_{2}^{2}=E^{i}_{k}(1)+2E^{i}_{k}(2)+E^{i}_{k}(3),$$
where
\begin{align*}
\begin{cases}
E^{i}_{k}(1)&=\mathbb{E}\Bigg\|\dfrac{\Big[f^{i}\big(x^i_{k}+c_{k}\bigtriangleup_{k}^{i}\big)-f^{i}\big(x^i_{k}-c_{k}\bigtriangleup_{k}^{i}\big)\Big]\big[\bigtriangleup_{k}^{i}\big]^{-}}{2c_{k}}\Bigg\|_{2}^{2},\\
E^{i}_{k}(2)&=\mathbb{E}\Bigg[\bigg\langle \dfrac{\Big[f^{i}\big(x^i_{k}+c_{k}\bigtriangleup_{k}^{i}\big)-f^{i}\big(x^i_{k}-c_{k}\bigtriangleup_{k}^{i}\big)\Big]\big[\bigtriangleup_{k}^{i}\big]^{-}}{2c_{k}},\dfrac{\epsilon_{k}^{i}\big[\bigtriangleup_{k}^{i}\big]^{-}}{2c_{k}} \bigg\rangle \Bigg],\\
E^{i}_{k}(3)&=\mathbb{E}\Bigg\|\dfrac{\epsilon_{k}^{i}\big[\bigtriangleup_{k}^{i}\big]^{-}}{2c_{k}}\Bigg\|_{2}^{2}.
\end{cases}
\end{align*}

By the Lipschitz property of $f^{i}(\cdot)$ in Remark \ref{Rem2},  $E^{i}_{k}(1)\leqslant L^2$. For the term $E^{i}_{k}(2)$, it follows from Hypothesis \ref{Hyp3}(c) and the Schwartz inequality that $\big|E^{i}_{k}(2)\big|\leqslant L\mathbb{E}\Bigg\|\dfrac{\epsilon_{k}^{i}\big[\bigtriangleup_{k}^{i}\big]^{-}}{2c_{k}}\Bigg\|_{2}\leqslant \dfrac{\sqrt{m}beL}{2c_{k}}$. Again by the mutual independence of $\{\varepsilon_k^i\}$ and $\{\bigtriangleup_{k}^{i}\}$ in Hypothesis \ref{Hyp3}(c)(d), we obtain
\begin{align}\label{A-3}
E^{i}_{k}(3)&=\frac{1}{4c_{k}^{2}}\mathbb{E}\big\|\epsilon^{i}_{k}\big\|_{2}^{2}\mathbb{E}\big\|[\bigtriangleup_{k}^{i}]^{-}\big\|_{2}^{2}\leqslant \dfrac{mb^{2}e^{2}}{4c_{k}^{2}}.
\end{align}
Therefore, $\mathbb{E}\big\|d_{k}^{i}\big\|_{2}^{2}\leqslant \Big(L+\dfrac{\sqrt{m}be}{2c_{k}}\Big)^{2}$, and thus, the second part of Lemma \ref{Lem5} is proved.
\end{itemize}

\section{Proof of Lemma \ref{Lem6}}

For all $i\in \mathcal{N}$, $k\geqslant 0$, define the error between $\xi^{i}_{k+1}$ and $x^{i}_{k}$ as
\begin{align}\label{B-1}
p^{i}_{k+1}\triangleq\xi^{i}_{k+1}-x^{i}_{k}=\xi^{i}_{k+1}-\sum_{j=1}^{n}w^{ij}_{k}\xi^j_{k}.
\end{align}

By Lemma \ref{Lem2}(b) and the fact that $X_{i}$ is a bounded closed convex set, we have
\begin{align}\label{B-2}
\big  \|p^i_{k+1}\big\|_{2}=\bigg\|P_{X}\big(\sum_{j=1}^{n}w^{ij}_{k}\xi^{i}_{k}-\iota_{k}d^{i}_{k}\big)-\sum_{j=1}^{n}w^{ij}_{k}\xi^j_{k}\bigg\|_{2}\leqslant\iota_{k}\big\|d^{i}_{k}\big\|_{2}.
\end{align}

Distributed subgradient-free algorithm with two-sided randomized differences in Algorithm \ref{Alg1} can be reformulated as follows:
\begin{align}\label{B-3}
\xi^{i}_{k+1}=\sum_{j=1}^{n} \big[\Psi_{k,0}\big]_{ij}\xi^{j}_{0}+p^{i}_{k+1}+\sum_{s=1}^{k}\sum_{j=1}^{n} \big[\Psi_{k,s}\big]_{ij}p^j_{s},
\end{align}
where $\Psi_{k,s}$ is defined in Lemma \ref{Lem4}. By Hypothesis \ref{Hyp2}, $W_{k}$ is a doubly stochastic matrix and hence, $\Psi_{k,0}$ is doubly stochastic. From (\ref{B-3}), the following equality holds:
\begin{align}\label{B-4}
\bar{\xi}_{k+1}=\frac{1}{n}\sum_{i=1}^{n}\xi^{i}_{0}+\frac{1}{n}\sum_{s=1}^{k+1}\sum_{j=1}^{n}p^j_{s}
\end{align}

From (\ref{B-3}) and (\ref{B-4}), for $\forall i\in \mathcal{N}$, we have the following inequalities
\begin{align}\label{B-5}
\big\|\xi^{i}_{k+1}-\bar{\xi}_{k+1}\big\|_2&\leqslant\sum_{j=1}^{n}\Big |\big [\Psi_{k,0}\big]_{ij}-\frac{1}{n}\Big|\big\|\xi^{j}_{0}\big\|_{2}+\big\|p^{i}_{k+1}\big\|_{2}
+\frac{1}{n}\sum_{j=1}^{n}\big\|p^j_{k+1}\big\|_{2}\notag\\&+\sum_{s=1}^{k}\sum_{j=1}^{n} \Big| \big[\Psi_{k,s}\big]_{ij}-\frac{1}{n}\Big|\big\|p^j_{s}\big\|_{2}.
\end{align}

By taking the expectation to both sides of \eqref{B-5} and recalling Lemma \ref{Lem4} and \eqref{B-2}, we obtain
\begin{align}\label{B-6}
\mathbb{E}\big\|\xi^{i}_{k+1}-\bar{\xi}_{k+1}\big\|_{2}&\leqslant n\lambda\beta^{k}\mathbb{E}\left[\max_{1\leqslant j\leqslant
	n}\big\|\xi^{j}_{0}\big\|_{2}\right]+\iota_{k}\mathbb{E}\big\|d^{i}_{k}\big\|_{2}+\frac{\iota_{k}}{n}\sum_{j=1}^{n}\mathbb{E}\big\|d^{j}_{k}\big\| _{2}
\notag\\&+\lambda\sum_{s=1}^{k}\beta^{k-s}\sum_{j=1}^{n}\iota_{s-1}\mathbb{E}\big\|d^j_{s-1}\big\|_{2}.
\end{align}

By Lemma \ref{Lem5}(a),  $\mathbb{E}\big\|d^{i}_{k}\big\|_{2}\leqslant  L+\frac{\sqrt{m}be}{2c_{k}}$, and hence,
\begin{align}\label{B-7}
\mathbb{E}\big\|\xi^{i}_{k+1}-\bar{\xi}_{k+1}\big\|_{2} &\leqslant n\lambda\beta^{k}\mathbb{E}\left[\max_{1\leqslant j\leqslant
	n}\big\|\xi^{j}_{0}\big\|_{2}\right]+2\iota_{k}\Big(L+\frac{\sqrt{m}be}{2c_{k}}\Big)\notag\\&+\lambda n\sum_{s=1}^{k}\iota_{s-1}\beta^{k-s}\Big(L+\frac{\sqrt{m}be}{2c_{s-1}}\Big).
\end{align}

Hypothesis \ref{Hyp4}(a) implies $\lim_{k\rightarrow\infty}\iota_{k}=0$.  By Hypothesis \ref{Hyp4}(c), $\lim_{k\rightarrow\infty}\frac{\iota_{k}}{c_{k}}=0$. By Lemma 3.1 in \cite{ram2010distributed} and $0<\beta<1$, we obtain
$\lim_{k\rightarrow\infty}\sum_{s=1}^{k}\iota_{s-1}\beta^{k-s}=0$. By Lemma \ref{Lem7}, we have $\lim_{k\rightarrow\infty}\sum_{s=1}^{k}\frac{\iota_{s-1}}{c_{s-1}}\beta^{k-s}=0$.
Therefore,
\begin{align}\label{B-8}
\lim_{k\rightarrow\infty}\mathbb{E}\big\|\xi^{i}_{k+1}-\bar{\xi}_{k+1}\big\|_{2}=0,\;\forall
i\in \mathcal{N}.
\end{align}
Thus, the proof is completed.

\section{Proof of Lemma \ref{Lem8}}
\begin{itemize}
\item[(a)] According to Lemma \ref{Lem1}:
\begin{align}\label{C-1}
f^{i}\big(x^i_{k}+c_{k}\bigtriangleup_{k}^{i}\big)-f^{i}\big(x^i_{k}-c_{k}\bigtriangleup_{k}^{i}\big)\in \big\langle \partial f^{i}\big(x^i_{k}+\theta^{i} c_{k}\bigtriangleup_{k}^{i}\big), \; 2c_{k}\bigtriangleup_{k}^{i}\big\rangle,
\end{align}
where $\theta^{i}\in[-1,1]$ is a constant. Therefore, there exists $\varsigma^{i} \in  \partial f^{i}\big(x^i_{k}+\theta^{i} c_{k}\bigtriangleup_{k}^{i}\big)$ such that
\begin{align}\label{C-2}
f^{i}\big(x^i_{k}+c_{k}\bigtriangleup_{k}^{i}\big)-f^{i}\big(x^i_{k}-c_{k}\bigtriangleup_{k}^{i}\big)=\big\langle \varsigma^{i}, \; 2c_{k}\bigtriangleup_{k}^{i}\big\rangle.
\end{align}
By taking conditional expectation of $\big\langle d^{i}_{k},x^{i}_{k}-\xi^*\big\rangle$ with respect to the $\sigma$-algebra $F_{k}$ and noticing (\ref{C-1}), we obtain the following inequality:
\begin{align}\label{C-3}
\mathbb{E}\big[\big\langle d^{i}_{k}, x^{i}_{k}-\xi^*\big\rangle \big | F_{k}\big]= G^{i}_{k}(1)+\frac{1}{2c_{k}}G^{i}_{k}(2),
\end{align}
where
\begin{align*}
\begin{cases}
G^{i}_{k}(1)&=\mathbb{E}\Big[ (\varsigma^{i})^{\top}\bigtriangleup_{k}^{i}\big[\bigtriangleup^{i}_{k}\big]^{-\top}(x^{i}_{k}-\xi^{*})\big | F_{k}\Big]\\
G^{i}_{k}(2)&=\mathbb{E}\Big[ \big\langle\epsilon_{i}^{k}\big[\bigtriangleup^{i}_{k}\big]^{-},\;x^{i}_{k}-\xi^{*}\big\rangle \big | F_{k}\Big].
\end{cases}
\end{align*}
For $ G^{i}_{k}(1)$, it can be further formulated as follows:
\begin{align}\label{C-4}
G^{i}_{k}(1)&=\mathbb{E}\Big[ (\varsigma^{i})^{\top}\Big(\bigtriangleup_{k}^{i}\big[\bigtriangleup^{i}_{k}\big]^{-\top}-I\Big)(x^{i}_{k}-\xi^{*})\big | F_{k}\Big]+\mathbb{E}\Big[\big\langle \varsigma^{i}, \;x^{i}_{k}-\xi^*\big\rangle\big | F_{k}\Big].
\end{align}

By Definition \ref{Def1} and Hypothesis \ref{Hyp1}(a)(b), we obtain
\begin{align}\label{C-5}
&\mathbb{E}[\big\langle \varsigma^{i},\;x^{i}_{k}-\xi^*\big\rangle|F_k]
\notag\\=&\mathbb{E}[\big\langle \varsigma^{i},\;x^{i}_{k}+\theta^{i}c_{k}\bigtriangleup_{k}^{i}-\theta^{i}c_{k}\bigtriangleup_{k}^{i}-\xi^*\big\rangle|F_k]
\notag\\
\geqslant & \mathbb{E}[f^{i}\big(x^i_{k}+\theta^{i}c_{k}\bigtriangleup_{k}^{i}\big)-f^{i}\big(\xi^*\big)|F_k]-\big |c_{k} \big | L\mathbb{E}\big\|\theta^{i}\bigtriangleup_{k}^{i}\big\|_{2}
 \notag\\\geqslant &
f^{i}\big(\bar{\xi}_{k}\big)-f^{i}\big(\xi^*\big)+\mathbb{E}[f^{i}\big(x^i_{k}+\theta^{i}c_{k}\bigtriangleup_{k}^{i}\big)-f^{i}\big(\bar{\xi_{k}}\big)|F_k]-\big |c_{k}\big | L\mathbb{E}\big\|\theta^{i} \bigtriangleup_{k}^{i}\big\|_{2}\notag\\\geqslant &
f^{i}\big(\bar{\xi}_{k}\big)-f^{i}\big(\xi^*\big)-L\big\|x^{i}_{k}-\bar{\xi}_{k}\big\|_{2}-2\big|c_{k}\big |L\mathbb{E}\big\|\theta^{i}\bigtriangleup_{k}^{i}\big\|_{2}
\notag\\\geqslant&
f^{i}\big(\bar{\xi}_{k}\big)-f^{i}\big(\xi^*\big)-L\big\|x^{i}_{k}-\bar{\xi}_{k}\big\|_{2}-2c_{k}L\mathbb{E}\big\|\bigtriangleup_{k}^{i}\big\|_{2}
\end{align}
and
\begin{align}
\nonumber&\left|\mathbb{E}\Big[ (\varsigma^{i})^{\top}\Big(\bigtriangleup_{k}^{i}\big[\bigtriangleup^{i}_{k}\big]^{-\top}-I\Big)(x^{i}_{k}-\xi^{*})\big | F_{k}\Big]\right|\\
\nonumber=&\left|\mathbb{E}\Big[ (\varsigma^{i}-\!\!\partial f^{i}\big(x^i_{k}\big))^{\top}\Big(\bigtriangleup_{k}^{i}\big[\bigtriangleup^{i}_{k}\big]^{-\top}\!\!-\!\!I\Big)(x^{i}_{k}\!\!-\!\!\xi^{*})\big | F_{k}\Big]\right|\\
\leqslant&2L.\label{C-6}
\end{align}
where $L$ is a positive constant and the last inequality takes place by the boundedness of $\partial f(\cdot)$ according to  Hypothesis \ref{Hyp1}(a)(b).

For $G^{i}_{k}(2)$, by $ F_{k}=\sigma\big\{x^{i}_{k},x^{i}_{k-1},\cdots,x^{i}_{0},i=1,\cdots,n;~\epsilon^{i}_{k-1},\epsilon^{i}_{k-2},\cdots, \epsilon^{i}_{0},i=1,\cdots,n;~\bigtriangleup^{i}_{k-1},\bigtriangleup^{i}_{k-2},\cdots,\bigtriangleup^{i}_{0},i=1,\cdots,n\big\}$, we have
\begin{align}\label{C-7}
G^{i}_{k}(2)= \big\langle\mathbb{E}\big[\epsilon^{i}_{k}\big[\bigtriangleup^{i}_{k}\big]^{-}\big| F_{k}\big],\;x^{i}_{k}-\xi^{*}\big\rangle= \big\langle\mathbb{E}\big[\epsilon^{i}_{k}\big[\bigtriangleup^{i}_{k}\big]^{-}\big],\;x^{i}_{k}-\xi^{*}\big\rangle.
\end{align}
From the above equality and by the mutual independence of $\varepsilon_k^i$ and $\bigtriangleup^{i}_{k}$ in Hypothesis \ref{Hyp3}(c)(d), it follows that
\begin{align}\label{C-8}
G^{i}_{k}(2)=\big\langle\mathbb{E}\big[\epsilon^{i}_{k}\big]\mathbb{E}\big[
 \bigtriangleup^{i}_{k}\big]^{-}],\;x^{i}_{k}-\xi^{*}\big\rangle=0.
\end{align}

Combining \eqref{C-5}, \eqref{C-6} and  \eqref{C-8} with  \eqref{C-3} gives
\begin{align}\label{C-9}
\mathbb{E}\big[\big\langle d^{i}_{k},\;x^{i}_{k}-\xi^*\big\rangle \big|F_{k}\big] &\geqslant f^{i}\big(\bar{\xi}_{k}\big)-f^{i}\big(\xi^*\big)-L\big\|x^{i}_{k}-\bar{\xi}_{k}\big\|_{2}-2c_{k}L\mathbb{E}\big\|\bigtriangleup_{k}^{i}\big\|_{2}-2L.
\end{align}

\item[(b)] The second part of the lemma can be proved by taking the expectation to both sides of (\ref{C-9}).  Thus, the proof completed.
\end{itemize}
\section{Proof of Lemma \ref{Lem9}}
By Lemma \ref{Lem2}(b), it holds that $\big\|\xi^i_{k+1}-\xi^*\big\|_{2}^2\leqslant \big\|\hat{\xi}^{i}_{k+1}-\xi^{*}\big\|_{2}^2$. Then, from (\ref{3-3-1}), it follows that
\begin{align}\label{D-1}
\big\|\xi^i_{k+1}-\xi^*\big\|_{2}^2= \big\|x^{i}_{k}-\xi^* \big\|_{2}^2 +
\iota^2_k
\big\|d^{i}_{k}\big\|_{2}^2-2\iota_{k}\big\langle d^{i}_{k},\;x^{i}_{k}-\xi^*\big\rangle.
\end{align}

By taking the conditional expectation to both sides of \eqref{D-1}, we obtain that, for all $k=0,1,2,\ldots$,
\begin{align}\label{D-2}
\mathbb{E}\Big[\big\|\xi^i_{k+1}\!\!-\!\!\xi^*\big\|_{2}^{2}\big|F_{k}\Big] &\!=\!\mathbb{E}\Big[\big\|x^{i}_{k}\!\!-\!\!\xi^{*}\big \|_{2}^{2} \big|F_{k}\Big]\!\!+\!\!\iota^{2}_{k}\mathbb{E}\Big[\big\|d^{i}_{k}\big\|_{2}^2\big|F_{k}\Big]\!\!-\!\!2\iota_{k}\mathbb{E}\Big[\big\langle d^{i}_{k},\; x^{i}_{k}\!\!-\!\!\xi^*\big\rangle \big|F_{k}\Big].
\end{align}

By the double stochasticity of matrix $W(k)$ in Hypothesis \ref{Hyp2}(b), we have
\begin{align}
\sum_{i=1}^{n} \mathbb{E}\Big[\big\|x^{i}_{k}-\xi^{*}\big\|_{2}^{2} \big|F_{k}\Big]&=   \sum_{i=1}^{n} \mathbb{E}\Big[\big\|\sum_{j=1}^{n}w^{ij}_{k}\xi^{j}_{k}-\xi^{*}\big\|_{2}^{2} \big|F_{k}\Big]
\leqslant \sum_{i=1}^{n} \big\|\xi^i_{k}-\xi^{*}\big\|_{2}^2,\label{D-3}\\
\sum_{i=1}^{n} \mathbb{E}\Big[\big\|x^{i}_{k}-\bar{\xi}_{k}\big\|_{2} \big|F_{k}\Big]&= \sum_{i=1}^{n} \mathbb{E}\Big[\big\|\sum_{j=1}^{N}w^{ij}_{k}\xi^{j}_{k}-\bar{\xi}_{k}\big\|_{2} \big|F_{k}\Big]\leqslant \sum_{i=1}^{n} \big\|\xi^i_{k}-\bar{\xi}_{k}\big\|_{2}\label{D-4}.
\end{align}

Then it follows from Lemma \ref{Lem8} that
\begin{align}\label{D-5}
\sum_{i=1}^{n}\mathbb{E}\Big[\big\|\xi^i_{k+1}-\xi^* \big\|_{2}^2\big|F_{k}\Big]&\leqslant\sum_{i=1}^{n}\Big[\big \|\xi^i_{k}-\xi^*\big\|_{2}^2+B^{i}_{k}(1)+B^{i}_{k}(2)+B^{i}_{k}(3)+B^{i}_{k}(4)-C^{i}_{k}\Big],
\end{align}
where $B^{i}_{k}(1)=\iota^{2}_{k}\mathbb{E}\Big[\big\|d^{i}_{k}\big\|_{2}^2\big|F_{k}\Big]$, $ B^{i}_{k}(2)=2\iota_{k}L\big\|\xi^{i}_{k}-\bar{\xi}_{k}\big\|_{2}$, $B^{i}_{k}(3)=4\iota_{k}c_{k}L\mathbb{E}\big\|\bigtriangleup_{k}^{i}\big\|_{2}$, $B^{i}_{k}(4)=4\iota_{k}L$,  $C^{i}_{k}=2\iota_{k}\big[f^{i}\big(\bar{\xi}_{k}\big)-f^{i}\big(\xi^*\big)\big]$.

According to Hypothesis \ref{Hyp4} and Lemma \ref{Lem5},  $\sum_{k=1}^{\infty}B^{i}_{k}(1)<\infty$ and $\sum_{k=1}^{\infty}B^{i}_{k}(4)<\infty$.
By Lemma \ref{Lem7}, $\sum_{k=1}^{\infty}B^{i}_{k}(2)<\infty$. By Hypothesis \ref{Hyp3}(a) and \ref{Hyp4}(c), we have $\sum_{k=1}^{\infty}B^{i}_{k}(3)<\infty$. Therefore, $\sum_{k=1}^{\infty}\sum_{i=1}^{n}[B^{i}_{k}(1)+B^{i}_{k}(2)+B^{i}_{k}(3)+B^{i}_{k}(4)]<\infty$.
Further, by noticing $\sum_{i=1}^{n}(f^{i}(\bar{\xi}_{k})-f^{i}(\xi^{*}))\geqslant 0$ and Lemma \ref{Lem3}, the sequence $\sum_{i=1}^{n} \big\|\xi^i_{k}-\xi^*\big\|_{2}^{2}$ converges to a non-negative random variable with probability 1.
Thus, the conclusion follows.
\section{Proof of Lemma \ref{Lem10}}
By taking expectation to both sides of \eqref{D-1}, we obtain
\begin{align}\label{E-1}
    \mathbb{E}\big\|\xi^i_{k+1}-\xi^*\big\|_{2}^{2} &\leqslant \mathbb{E}\big\|x^{i}_{k}-\xi^{*} \big\|_{2}^{2}+\iota^{2}_{k}\mathbb{E}\big\|d^{i}_{k}\big\|_{2}^2-2\iota_{k}\mathbb{E}\big[\big\langle d^{i}_{k},\; x^{i}_{k}-\xi^*\big\rangle\big].
\end{align}

By the double stochasticity of matrix $W(k)$ given in Hypothesis \ref{Hyp2}(b), we have the following inequalities
\begin{align}
\sum_{i=1}^{n} \mathbb{E}\big\|x^{i}_{k}-\xi^{*}\big\|_{2}^{2} &=   \sum_{i=1}^{n} \mathbb{E}\Big\|\sum_{j=1}^{n}w^{ij}_{k}\xi^{j}_{k}-\xi^{*}\Big\|_{2}^{2}\leqslant \sum_{i=1}^{n} \mathbb{E}\big\|\xi^i_{k}-\xi^{*}\big\|_{2}^2,\label{E-2}\\\sum_{i=1}^{n} \mathbb{E}\big\|x^{i}_{k}-\bar{\xi}_{k}\big\|_{2}&= \sum_{i=1}^{n} \mathbb{E}\Big\|\sum_{j=1}^{N}w^{ij}_{k}\xi^{j}_{k}-\bar{\xi}_{k}\Big\|_{2}\leqslant \sum_{i=1}^{n}\mathbb{E} \big\|\xi^i_{k}-\bar{\xi}_{k}\big\|_{2}.\label{E-3}
\end{align}
By taking summation of both sides of \eqref{E-1} for $k=1,2,\ldots T$ and $i=1,2,\ldots n$ and noticing \eqref{E-2}, \eqref{E-3} and Lemma \ref{Lem8}, we have

\begin{align}\label{E-4}
\sum_{s=1}^{k}\sum_{i=1}^{n}\mathbb{E} \big\|\xi^i_{s+1}-\xi^{*}\big\|_{2}^{2}&\leqslant \sum_{s=1}^{k}\sum_{i=1}^{n}\mathbb{E} \big\|\xi^i_{s}-\xi^{*}\big\|_{2}^{2}+\sum_{s=1}^{k}\sum_{i=1}^{n}\iota^{2}_{k}\mathbb{E}\big\|d^{i}_{s}\big\|_{2}^2+2L\sum_{s=1}^{k}\sum_{i=1}^{n}\iota_{s}\mathbb{E}\big\|\xi^{i}_{s}-\bar{\xi}_{s}\big\|_{2}\notag\\&+4L\sum_{s=1}^{k}\sum_{i=1}^{n}\iota_{s}c_{s}\mathbb{E}\big\|\bigtriangleup_{s}^{i}\big\|_{2}+4nL\sum_{s=1}^{k}\iota_{s}-2\sum_{s=1}^{k}\sum_{i=1}^{n}\iota_{s}\mathbb{E}\Big[f^{i}\big(\bar{\xi}_{s}\big)-f^{i}\big(\xi^*\big)\Big].
\end{align}
Therefore,
\begin{align}\label{E-5}
\sum_{i=1}^{n}\mathbb{E} \big\|\xi^i_{k+1}-\xi^{*}\big\|_{2}^{2}&\leqslant \sum_{s=1}^{k}\sum_{i=1}^{n}\iota^{2}_{s}\mathbb{E}\big\|d^{i}_{s}\big\|_{2}^2+2L\sum_{s=1}^{k}\sum_{i=1}^{n}\iota_{s}\mathbb{E}\big\|\xi^{i}_{s}-\bar{\xi}_{s}\big\|_{2}+4L\sum_{s=1}^{k}\sum_{i=1}^{n}\iota_{s}c_{s}\mathbb{E}\big\|\bigtriangleup_{s}^{i}\big\|_{2}\notag\\&+4nL\sum_{s=1}^{k}\iota_{s}.
\end{align}
Noticing that $\iota_{s}=\frac{1}{s^{1+\epsilon}}$, $c_{s}=\frac{1}{s^{\delta}}$, and $ \frac{1}{2}+\epsilon>\delta>0 $. By Lemma \ref{Lem5}, for the first term on the right hand side of \eqref{E-5}, we have
\begin{align}\label{E-6}
\sum_{s=1}^{k}\sum_{i=1}^{n}\iota^{2}_{s}\mathbb{E}\big\|d^{i}_{s}\big\|_{2}^2\leqslant n\sum_{s=1}^{k}\iota^{2}_{k}\Big(L+\dfrac{\sqrt{m}be}{2c_{k}}\Big)^{2}\leqslant M_{0}\sum_{k=1}^{s}\frac{\iota^{2}_{s}}{c^{2}_{s}}\leqslant \frac{M_{1}}{k^{1+2\epsilon-2\delta}}.
\end{align}
Since $X_{i}$ is bounded in $\mathcal{R}^{m}$, for $x\in X_{i}$, there exists a constant $M_{x}$ such that $\big\|x\big\|_{2}\leqslant M_{x}$. For the second term on the right hand side of \eqref{E-5}, we have
\begin{align}\label{E-7}
2L\sum_{s=1}^{k}\sum_{i=1}^{n}\iota_{s}\mathbb{E}\big\|\xi^{i}_{s}-\bar{\xi}_{s}\big\|_{2}\leqslant 4nLM_{x}\sum_{s=1}^{k}\iota_{s}\leqslant\dfrac{M_{2}}{k^{\epsilon}}.
\end{align}
According to Hypothesis \ref{Hyp3}, for  the third term on the right hand side of \eqref{E-5}, we have
\begin{align}\label{E-8}
4L\sum_{s=1}^{k}\sum_{i=1}^{n}\iota_{s}c_{s}\mathbb{E}\big\|\bigtriangleup_{s}^{i}\big\|_{2}\leqslant  \frac{M_{3}}{k^{\epsilon+\delta}}.
\end{align}
For  the last term on the right hand side of \eqref{E-5}, we have
\begin{align}\label{E-9}
4nL\sum_{s=1}^{k}\iota_{k}\leqslant\dfrac{M_{4}}{k^{\epsilon}}.
\end{align}
$M_{0}, M_{1}\ldots,M_{4}$ are positive constants in the above inequalities, we have
\begin{align}\label{E-10}
\sum_{i=1}^{n}\mathbb{E} \big\|\xi^i_{k+1}-\xi^{*}\big\|_{2}^{2}\leqslant  \frac{M_{1}}{k^{1+2\epsilon-2\delta}}+\dfrac{M_{2}}{k^{\epsilon}}+ \frac{M_{3}}{k^{\epsilon+\delta}}+\dfrac{M_{4}}{k^{\epsilon}}.
\end{align}
\bibliographystyle{IEEEtran}

\bibliography{IEEEabrv,Distributed_KW_Algorithm}

\end{document}